\documentclass[11pt,reqno]{amsart}
\usepackage[T1]{fontenc}
\usepackage{lmodern,microtype,amsmath,amssymb,amsthm,mathtools,booktabs,array,longtable}
\usepackage[margin=1.05in]{geometry}
\usepackage{cite}
\usepackage{xurl,hyperref}
\DeclareUrlCommand\path{\urlstyle{tt}\def\UrlBreaks{\do\/\do\-}\def\UrlBigBreaks{}}
\hypersetup{hidelinks,pdftitle={Large-Scale Autonomous Discovery of Kissing Number Constructions}}
\makeatletter
\renewcommand{\@settitle}{%
  \begin{center}\normalfont\fontsize{13}{16}\selectfont\bfseries\@title\end{center}%
}
\renewcommand{\@setauthors}{%
  \begingroup
  \trivlist\centering\footnotesize
  \@topsep30\p@\relax
  \advance\@topsep by -\baselineskip
  \item\relax\MakeUppercase{\PaperAuthorBlock}%
  \endtrivlist
  \endgroup
}
\AddToHook{cmd/@setauthors/after}{%
  \par\nobreak\smallskip
  {\centering\footnotesize\normalfont\PaperAffiliation\par}%
  \nobreak\smallskip
  {\centering\footnotesize\normalfont\PaperCorrespondence\par}%
  \global\let\addresses\@empty
}
\makeatother
\numberwithin{equation}{section}
\newtheorem{theorem}{Theorem}[section]
\newtheorem{proposition}[theorem]{Proposition}
\newtheorem{lemma}[theorem]{Lemma}
\newtheorem{corollary}[theorem]{Corollary}
\theoremstyle{definition}
\theoremstyle{remark}
\newcommand{\R}{\mathbb R}\newcommand{\Z}{\mathbb Z}\newcommand{\Q}{\mathbb Q}\newcommand{\F}{\mathbb F}
\newcommand{\wt}{\operatorname{wt}}\newcommand{\supp}{\operatorname{supp}}
\newcommand{\ip}[2]{\langle #1,#2\rangle}\newcommand{\norm}[1]{\lVert #1\rVert}
\newcommand{\eps}{\varepsilon}\newcommand{\qpath}[1]{\texttt{\nolinkurl{#1}}}

\newcommand{\AuthorNames}{Shuxing Yang, Rui Zhao, Junyao Wu, Yize Wang,
Fujia Chen, Kaihao Zhu, Wenhao Li, Zichen Li, Yaqi Li, Shenzhan Hong,
Yuang Pan, Junjie Yang, Taowen Deng, Jincheng Mi, Hongsheng Chen, Yihao Yang}
\author[Shuxing Yang et al.]{Shuxing Yang}
\author[]{Rui Zhao}
\author[]{Junyao Wu}
\author[]{Yize Wang}
\author[]{Fujia Chen}
\author[]{Kaihao Zhu}
\author[]{Wenhao Li}
\author[]{Zichen Li}
\author[]{Yaqi Li}
\author[]{Shenzhan Hong}
\author[]{Yuang Pan}
\author[]{Junjie Yang}
\author[]{Taowen Deng}
\author[]{Jincheng Mi}
\author[]{Hongsheng Chen\textsuperscript{*}}
\author[]{Yihao Yang\textsuperscript{*}}
\newcommand{\PaperAuthorBlock}{%
Shuxing~Yang, Rui~Zhao, Junyao~Wu, Yize~Wang, Fujia~Chen, Kaihao~Zhu,\\
Wenhao~Li, Zichen~Li, Yaqi~Li, Shenzhan~Hong, Yuang~Pan, Junjie~Yang,\\
Taowen~Deng, Jincheng~Mi, Hongsheng~Chen\textsuperscript{*}, and Yihao~Yang\textsuperscript{*}}
\newcommand{\PaperAffiliation}{College of Information Science and Electronic Engineering, Zhejiang University\\
Qiushi Engine Team}
\address{\PaperAffiliation}
\newcommand{\PaperCorrespondence}{\textsuperscript{*}Corresponding authors:
Hongsheng Chen (\href{mailto:hansomchen@zju.edu.cn}{hansomchen@zju.edu.cn});
Yihao Yang (\href{mailto:yangyihao@zju.edu.cn}{yangyihao@zju.edu.cn}).}
\hypersetup{pdfauthor={\AuthorNames}}

\title[Autonomous discovery of kissing constructions]{Large-Scale Autonomous Discovery of Kissing Number Constructions}
\date{25 September 2026}
\begin{document}
\begin{abstract}
The kissing-number problem is a classical problem in discrete geometry
whose exact solution is known in only a few dimensions. Recent
artificial-intelligence approaches have begun to discover improved
configurations through large-scale numerical and combinatorial search,
but converting such searches into general mathematical constructions and
rigorous proofs remains challenging. Here we use Qiushi Engine, an autonomous
multi-agent research system, to investigate kissing numbers and obtain
new lower bounds in nineteen dimensions: $25$, $27$, $32$--$39$, $43$,
$45$, and $49$--$55$. The resulting constructions arise from distinct
structural mechanisms, including coordinated motions of contact layers,
labelled direction reuse, joint support exchanges, signed-code
replacements, cross-shell lattice constructions, low-overlap lattice
isometries, and spherical-design moment certificates. They yield sharp
capacities for parameterized signed-code models, deterministic image-union
guarantees, and exact section and projection counts controlled by embedded
root systems and anchor-graph statistics. These methods
yield, among others, $K(25)\ge197580$, $K(27)\ge201567$,
$K(38)\ge591900$, $K(43)\ge2553792$, $K(45)\ge7380090$, and
$K(55)\ge53301140$. The autonomous system carried out the construction
searches, mathematical analysis and computational verification, while
each final result was reduced to explicit mathematical arguments and
independently checkable finite certificates. Our results illustrate how
autonomous research systems can move beyond optimization within fixed
formulations to discover new mathematical representations and
constructions at scale.
\end{abstract}
\maketitle
\section{Introduction}\label{sec:introduction}
The kissing-number problem is one of the oldest problems in discrete
geometry. In 1694, Isaac Newton and David Gregory discussed whether
12 or 13 equal spheres could simultaneously touch a central sphere in
three dimensions. The answer, 12, was proved more than two centuries
later~\cite{musin}. Yet the apparent simplicity of the question is
deceptive. More than three hundred years after its formulation, exact
kissing numbers are known only in dimensions $1,2,3,4,8$ and $24$.
In the exceptional dimensions $8$ and $24$, the optimal configurations
arise from the $E_8$ and Leech lattices, with 240 and 196560 contacts,
respectively~\cite{conway1999sphere,musin}. In most other dimensions,
substantial gaps remain between the best constructions and known upper
bounds.

Progress on lower bounds has historically relied on explicit mathematical
structure. Lattices and error-correcting codes provide large families of
vectors whose pairwise interactions can be controlled simultaneously
\cite{conway1999sphere,macwilliams}. These ideas have led to
laminated-lattice constructions, lattice sections and extensions,
code-based spherical configurations, and lifting methods that combine
a high-dimensional mother configuration with additional coordinates
\cite{laminated,antipode,ers,cjkt,kallal}. Such constructions can contain
hundreds of thousands or even tens of millions of vectors. Their strength,
however, is also a source of difficulty: once a configuration is highly
constrained, improving it may require reorganizing an entire collection
of points rather than inserting a single new one. The relevant search
spaces grow combinatorially, while the geometric degrees of freedom that
permit improvement can be difficult to identify from the original
representation.

Artificial intelligence has recently begun to alter how these spaces are
explored. AlphaEvolve, a general-purpose coding agent, discovered a
593-point kissing configuration in dimension 11, improving the previously
known lower bound of 592~\cite{alphaevolve}. PackingStar formulated the kissing-number
problem as a two-player matrix-completion game and used game-theoretic
reinforcement learning to search the space of Gram matrices, reporting
new constructions across dimensions 25--31 and uncovering thousands of
candidate geometric structures~\cite{packingstar-discovery}. EinsteinArena
allowed autonomous agents to submit solutions, inspect intermediate
results and build on discoveries made by other agents; this collective
process increased the dimension-11 lower bound from 593 to
604~\cite{einsteinarena}. The Station subsequently reported exact
604-point configurations and structural theorems developed in an
open-world multi-agent environment~\cite{station}. Together, these
developments demonstrate the scope of AI-guided geometric search and
collaborative mathematical discovery.

These advances also expose a deeper challenge. Finding an improved
configuration is only one part of mathematical discovery. Why does the
configuration exist, which structural degree of freedom makes it possible,
can the mechanism be generalized, and how can the result be converted into
a proof or an independently verifiable certificate? Computational
searches begin from a prescribed representation---for example,
coordinates, Gram matrices, codes or a fixed optimization objective.
Open-ended mathematical research additionally requires the ability to
change that representation when the current one becomes limiting, to
formulate new intermediate problems, connect apparently unrelated
structures, and develop arguments that certify the resulting
construction.

Autonomous scientific research systems provide a possible route towards
this broader form of discovery. Recent large-language-model-based agents
can carry out extended sequences of reasoning, computation and tool use,
while revising research strategies in response to intermediate evidence.
Qiushi Engine was previously used for end-to-end autonomous scientific
discovery on a physical optical platform~\cite{qiushi-optics}, where it
conducted long-horizon investigations that combined theoretical
reasoning, computation and experiment. Here we ask whether the same
paradigm can be extended to open mathematical research: can an autonomous
research system move beyond optimizing a fixed mathematical formulation
and instead discover new representations, constructions and proofs?

Here we report an autonomous investigation of the kissing-number problem
that yields new lower bounds in nineteen dimensions:
\[
25,\quad27,\quad32\text{--}39,\quad43,\quad45,\quad49\text{--}55.
\]
The resulting bounds include
\[
\begin{aligned}
K(25)&\ge197580,&K(27)&\ge201567,&K(38)&\ge591900,\\
K(43)&\ge2553792,&K(45)&\ge7380090,&K(55)&\ge53301140.
\end{aligned}
\]
Rather than arising from a single optimization procedure, these
improvements were obtained through a sequence of distinct structural
discoveries. They include coordinated motions of complete contact layers,
labelled liftings that separate direction deletion from reuse, joint
support exchanges that exploit shared deletion costs, signed-code
replacements generated through Hadamard transfers, compatible directions
drawn from a different lattice shell, enlarged lattice-line classes
combined through low-overlap isometric images, and spherical-design
moment identities and inequalities that determine or bound large lattice
sections without enumerating the ambient shell.

A common pattern connects these apparently different constructions.
Configurations that appear saturated in one representation retain
exploitable freedom in another. In dimension 25, moving an entire
552-point contact layer while preserving its internal Gram matrix
creates room for an additional point. In dimension 27, labelled lifting
separates direction reuse from deletion; with the first layer fixed,
jointly selecting second-layer directions and labels enlarges the final
two-layer construction. In dimensions 32--37 and 39, optimizing
collective exchanges and individual sign patterns exposes improvements
hidden by separate support-insertion scores and complete sign bundles.
In dimension 38, changing the search from the minimal Leech shell to a
second lattice shell produces 288 additional points. In dimensions 43
and 45, counting millions of minimal lattice vectors reduces to a finite
space of projection signatures whose multiplicities satisfy exact
spherical-design moments. A single enlarged line class in dimension 48,
combined with low-overlap isometric images, similarly generates new
constructions throughout dimensions 49--55.

Qiushi Engine carried out the construction searches, mathematical analysis
and computational verification. Each discovered mechanism is expressed
through explicit geometric or combinatorial conditions. Each reported
bound is supported by finite witnesses checked using exact integer and
rational arithmetic, lattice and code identities, rigorous interval
arithmetic, and finite combinatorial verification. In dimension 25,
outward-rounded interval arithmetic certifies the remaining strict
inequalities~\cite{johansson-arb}. These mathematical arguments and
reproducible certificates make the constructions independently
checkable.

These results suggest a role for AI in mathematics that extends beyond
accelerating numerical search. For difficult problems with enormous
structured search spaces, a central task may be to discover the
representation in which progress becomes possible. Autonomous research
systems that can alternate between exploration, structural interpretation
and rigorous verification may provide a route towards such discoveries
at scale.

\subsection{Main results and structural statements}

The kissing number $K(d)$ is the largest cardinality of a set
$X\subset S^{d-1}$ satisfying
\begin{equation}\label{eq:kissing}
 \ip{x}{y}\le\tfrac12\qquad(x\ne y).
\end{equation}
Equivalently, it is the largest number of nonoverlapping unit balls
that can touch a central unit ball in $\R^d$.

\begin{theorem}\label{thm:main}
For every pair $(d,N_d)$ in Table~\ref{tab:main}, there is a kissing
configuration of $N_d$ points in $\R^d$. In particular, $K(d)\ge N_d$.
\end{theorem}

\begin{longtable}{rrrrl}
\caption{The nineteen lower bounds, with the public comparison values
recorded in the accompanying catalogue on 24--27 September 2026.}
\label{tab:main}\\
\toprule
$d$&$N_d$&Comparison&Increase&Source\\\midrule
\endfirsthead
\toprule $d$&$N_d$&Comparison&Increase&Source\\\midrule\endhead
25&197580&197579&1&\cite{kissingnumbers}\\
27&201567&201566&1&\cite{lindow27}\\
32&347584&346944&640&\cite{brouwer}\\
33&363968&362048&1920&\cite{brouwer,echols}\\
34&384196&381124&3072&\cite{brouwer,echols}\\
35&409676&409548&128&\cite{brouwer}\\
36&484760&484568&192&\cite{brouwer}\\
37&498024&496232&1792&\cite{brouwer,echols}\\
38&591900&591612&288&\cite{kissingnumbers}\\
39&763668&756116&7552&\cite{kissingnumbers}\\
43&2553792&2545056&8736&\cite{crosssections}\\
45&7380090&7379838&252&\cite{crosssections}\\
49&52430156&52430140&16&\cite{kissingnumbers}\\
50&52458468&52458418&50&\cite{kissingnumbers}\\
51&52500930&52500816&114&\cite{kissingnumbers}\\
52&52585772&52585516&256&\cite{kissingnumbers}\\
53&52698696&52698222&474&\cite{kissingnumbers}\\
54&52923774&52922906&868&\cite{kissingnumbers}\\
55&53301140&53299730&1410&\cite{kissingnumbers}\\\bottomrule
\end{longtable}

For dimension 32, the comparison follows from Lysenst\o en's
1671-word code recorded in Brouwer's table~\cite{brouwer}:
$2^{17}+128\cdot1671+2\cdot32\cdot31=346944$.
The 1667-word code of Echols~\cite{echols} was the seed for our support search.

\subsection{Structural consequences}

Four structural principles connect the nineteen constructions.

\paragraph{Motion and lifting.}
A contact layer can move while preserving its Gram matrix; convexity
controls its interaction with the fixed equator
(Lemma~\ref{lem:chord-motion}). The remaining boundary constraints have
a finite exact test, and two cap changes are necessary and sufficient
for the specified dimension-25 insertion
(Propositions~\ref{prop:chord-finite} and~\ref{prop:motion-two-repairs}).
For labelled liftings, counting copies separately from deleted directions
exposes a reuse problem (Theorem~\ref{lem:lift}). Weighted incidence
bounds and the sharp simplex criterion relate its capacity to head
separation and tail dimension (Propositions~\ref{prop:weighted}
and~\ref{prop:simplex}).

\paragraph{Exchange and signed designs.}
The gain from an internally compatible insertion $Y$ is
$|Y|-|F(Y)|$, where $F(Y)$ is the union of its old conflicts.
This shared-deletion objective governs the support exchanges; for a
compatible candidate pool, its maximum is determined by a matching
(Propositions~\ref{prop:exchange} and~\ref{prop:exchange-deficiency}).
Complete sign bundles admit an exact boundary criterion, allowing
paired Hadamard transfer to replace a global compatibility problem by
a local design problem. The resulting codes include a sharp
eleven-coordinate construction and a family of $32t^2(t-1)$ vectors
in $4t$ coordinates attaining its transverse $2+2$ model capacity.
For twelve coordinates, arbitrary restrictions by complete pair types
also have an exact capacity formula
(Section~\ref{sec:signed}, Theorem~\ref{thm:type-boundary-capacity}).

\paragraph{Shell completion and image unions.}
The lattice minimum guarantees all new--old products in a shell-norm
window, leaving only the compatibility of the added directions to
be arranged (Corollary~\ref{cor:completion-window}). A multi-shell
criterion and a next-shell obstruction locate the remaining freedom
for the dimension-38 base (Propositions~\ref{prop:multi-shell}
and~\ref{prop:no-norm64}). In dimensions 49--55, retaining pure tails
forces the paired lifting heights and disjoint head assignment.
Thus the optimization in this model is exactly that of an image union,
not a sum of image sizes (Proposition~\ref{prop:paired-rigidity}).
Orbitwise and finite-pool bounds provide deterministic selection
guarantees (Propositions~\ref{prop:orbit-union} and~\ref{prop:finite-pool}).

\paragraph{Moment-certified section statistics.}
Complete signature domains and design moments can determine a target
statistic without determining every fibre multiplicity
(Proposition~\ref{prop:section-statistics}). For an extremal even
unimodular lattice of rank 48, an embedded $\sqrt3E_8$ forces eight
fixed coordinate-section counts, including the specified $D_5$ count
2553792 (Theorems~\ref{d43:thm:structure}
and~\ref{thm:e8-coordinate-sections}). The parent embedding supplies
information absent from the child's Gram matrix alone.
For rank-three sections of Gram $8I_3-2J_3$, the projected count is
exactly $7380720-9c$, where $c$ is the common-neighbour count of a
nonadjacent pair in a 1128-vertex anchor graph
(Theorem~\ref{thm:45-transfer}, Proposition~\ref{prop:45-anchor-graph}).
Its signed Gram identities restrict $c$; exhaustive counting gives
$\min c=70$ for the supplied anchor, yielding 7380090 points.

\subsection{Relation to earlier constructions}

Cohn--Jiao--Kumar--Torquato~\cite[Section 7]{cjkt} and
Kallal--Kan--Wang~\cite{kallal} developed Leech lifting methods.
The 496-point blocks used in our four-triangle example are from
PackingStar~\cite{packingstar,packingstar-data}.
Takhanov--Yun~\cite{takhanov-yun} exploit freedom in lifted blocks while
keeping the unlifted bulk fixed. Our 25-dimensional construction instead
moves retained equatorial points of the 197579-point configuration
in~\cite{kissingnumbers}. The current two-layer input in dimension 27
is Lindow's 201566-point contribution~\cite{lindow27}.

The code-based constructions use the layers of
Edel--Rains--Sloane~\cite{ers}, the dense codes of
Cheng--Sloane~\cite{chengsloane}, and constant-weight constructions
and tables~\cite{bsss,brouwer,echols}. Local replacement has a direct
precursor in the free-zone construction of
Takhanov--Yun~\cite[Section V]{signed-johnson}; the 51-block
weight-four design is classical~\cite{kalbfleisch-stanton}.
Here the new supports, transverse choices, signed replacements and
their boundary embeddings supply the improvements.

The empty-equator lift in dimension 38 and the 7069-line class in
dimension 48 are inputs from~\cite{kissingnumbers}.
For sections, we use the framework of Conway--Sloane
\cite{laminated,antipode}, the embedded parent geometry
of~\cite{dorofeev}, and the explicit lattice and section data
of~\cite{latticecatalogue,p48p-catalogue,crosssections}.
The new constructions specify the added shell lines, the enlarged
line class and image families, and the section statistics and witnesses
used in the stated bounds.

\subsection{Notation and organization}

We write $[n]=\{0,\ldots,n-1\}$ unless a finite input explicitly uses
one-based coordinates. The notation $A(n,\delta,w)$ denotes the largest
binary constant-weight code of length $n$, weight $w$ and minimum
Hamming distance $\delta$. For vectors of common squared norm $q$,
the kissing condition is $\ip{x}{y}\le q/2$ before normalization.
The main coordinate scales are collected in Table~\ref{tab:scales};
products between different shells are treated separately.

\begin{table}[ht]
\centering
\caption{Coordinate scales used in the constructions.}
\label{tab:scales}
\begin{tabular}{lrr}\toprule
Point set & Squared norm & Product threshold\\\midrule
Normalized configurations & $1$ & $1/2$\\
Dimension-25 motion & $4$ & $2$\\
Integer Leech minimal shell & $32$ & $16$\\
Local signed codes & $8$ & $4$\\
Dimension-38 added shell & $48$ & $24$\\
Rank-48 minimal shells & $6$ & $3$\\
Dimension-45 projected fibres & $23/4$ & $23/8$\\\bottomrule
\end{tabular}
\end{table}

The paper first treats lifting and motion, then support and signed-code
replacement, second-shell completion, and the seven block lifts from $P_{48p}$.
The final construction sections develop the two moment arguments.
Appendices specify the Leech coordinates, dense codes, an explicit Pless
neighbour, the common theta-series shell count, and deformation estimates.
Section~\ref{sec:certificates} connects
the proofs to the finite inputs and independent checking programs.
The accompanying catalogue records dated public comparisons separately
from the existence statements of Theorem~\ref{thm:main}.

Section~\ref{sec:autonomous-research} and the accompanying research
accounts describe Qiushi Engine's mathematical contributions and the
development of the constructions.

\section{Lifting and the cost of reusing a direction}\label{sec:lifting}

Let $C\subset S^{D-1}$ be a kissing configuration. Decompose the enlarged
space as $\R^{D+k}=\R^D\oplus\R^k$. We refer to the second component as
the \emph{tail}. Fix $\ell\ge1$ distinct tails $t_1,\ldots,t_\ell$ with
$0<\norm{t_i}<1$, and put $a_i=\sqrt{1-\norm{t_i}^2}$. Choose subsets
$A_i\subseteq C$, write $U=\bigcup_i A_i$, and let
$Y\subset S^{k-1}$ be a kissing configuration. Consider
\begin{align}
 Z&=\{(x,0):x\in C\setminus U\},\nonumber\\
 M&=\bigcup_{i=1}^{\ell}\{(a_i x,t_i):x\in A_i\},\label{eq:general-lift}\\
 P&=\{(0,y):y\in Y\}.\nonumber
\end{align}
Put $X=Z\cup M\cup P$. First-component norms distinguish the three
parts, and tails distinguish the copies in $M$, so
\begin{equation}\label{eq:lift-count}
 |Z\cup M\cup P|=|C|-|U|+\sum_i|A_i|+|Y|.
\end{equation}

\begin{theorem}[Labelled lifting criterion]\label{lem:lift}
The set in~\eqref{eq:general-lift} is a kissing configuration if and only if
\begin{align}
 a_i a_j\ip{x}{x'}+\ip{t_i}{t_j}&\le\frac12
 &&\text{whenever }(x,i)\ne(x',j),\ x\in A_i,\ x'\in A_j,
 \label{eq:lift-mixed}\\
 \ip{t_i}{y}&\le\frac12
 &&(i\text{ with }A_i\ne\varnothing,\ y\in Y).\label{eq:lift-pure}
\end{align}
\end{theorem}
\begin{proof}
Every point has norm one. Equator--equator and pure-tail pairs inherit
the inequalities for $C$ and $Y$. Equator--pure-tail pairs are orthogonal.
For an equator--mixed pair, the two original directions are different
because the equator excludes $U$; their inner product is at most
$a_i/2<1/2$. The remaining pairs have exactly the products in
\eqref{eq:lift-mixed} and~\eqref{eq:lift-pure}, which are therefore
both necessary and sufficient.
\end{proof}

Equivalently, the tail Gram matrix specifies a threshold matrix
\begin{equation}\label{eq:kappa}
 \kappa_{ij}=\frac{1/2-\ip{t_i}{t_j}}{a_i a_j},
\end{equation}
and the selected base directions must satisfy
$\ip{x}{x'}\le\kappa_{ij}$ for every distinct labelled pair.
This allows unequal latitudes, subject separately to~\eqref{eq:lift-pure}.

For equal tail norms, put $\norm{t_i}^2=h^2$ and
$a_i=a=\sqrt{1-h^2}$. The same-label condition is
\begin{equation}\label{eq:block-threshold}
 \ip{x}{x'}\le\frac{1/2-h^2}{1-h^2}\qquad(x\ne x',\ x,x'\in A_i),
\end{equation}
whereas two copies of the same direction require
\begin{equation}\label{eq:reuse-threshold}
 \ip{t_i}{t_j}\le h^2-\frac12\qquad(x\in A_i\cap A_j,\ i\ne j).
\end{equation}
These inequalities control block size and multiplicity, respectively.

\subsection{The incidence-capacity inequality}
For general tail norms define a graph $\Gamma$ on the labels by
\begin{equation}\label{eq:reuse-graph}
 ij\in E(\Gamma)\quad\Longleftrightarrow\quad
 a_i a_j+\ip{t_i}{t_j}\le\frac12\qquad(i\ne j).
\end{equation}
The labels assigned to one base direction form a clique in $\Gamma$.

\begin{proposition}[Weighted incidence bound]\label{prop:weighted}
Suppose~\eqref{eq:general-lift} is a kissing configuration with
$|A_i|\le\alpha_i$. If nonnegative numbers $\lambda_1,\ldots,\lambda_\ell$
satisfy
\begin{equation}\label{eq:clique-cover}
 \sum_{i\in J}\lambda_i\ge |J|-1
 \quad\text{for every nonempty clique }J\text{ of }\Gamma,
\end{equation}
then
\begin{equation}\label{eq:weighted-capacity}
 |X|\le |C|+|Y|+\sum_i\lambda_i\alpha_i.
\end{equation}
Equality holds exactly when every used direction's label clique is
tight in~\eqref{eq:clique-cover} and every label with
$\lambda_i>0$ reaches its capacity.
\end{proposition}
\begin{proof}
For $x\in U$, let $J(x)=\{i:x\in A_i\}$. The count gives
\[
 |X|-|C|-|Y|=\sum_{x\in U}(|J(x)|-1)
 \le\sum_{x\in U}\sum_{i\in J(x)}\lambda_i
 =\sum_i\lambda_i|A_i|\le\sum_i\lambda_i\alpha_i.
\]
Both inequalities are sums of nonnegative deficits. Their equality
conditions are precisely those stated.
\end{proof}

If $\Gamma$ has clique number $\omega$, the uniform choice
$\lambda_i=1-1/\omega$ is feasible. Thus
\begin{equation}\label{eq:clique-capacity}
 |X|\le |C|+|Y|+
 \left(1-\frac1\omega\right)\sum_i\alpha_i.
\end{equation}
Minimizing $\sum_i\lambda_i\alpha_i$ subject to
\eqref{eq:clique-cover} is a finite linear program; any rational feasible
choice bounds the prescribed family. The twelve tails below have eight
maximum cliques, all triangles; four disjoint triangles saturate the
uniform bound. Rankin's simplex argument~\cite{rankin} also bounds
multiplicity directly from the tail geometry.

\begin{proposition}[Multiplicity bound]\label{prop:multiplicity}
Suppose~\eqref{eq:general-lift} is a kissing configuration, all tails
have squared norm $h^2$, and $h^2<1/2$. The number of copies of any
original direction is at most
\begin{equation}\label{eq:max-multiplicity}
 r=\min\left\{k+1,\left\lfloor\frac1{1-2h^2}\right\rfloor\right\}.
\end{equation}
If $|A_i|\le\alpha_i$, its cardinality is consequently at most
\begin{equation}\label{eq:reuse-bound}
 |C|+|Y|+\left(1-\frac1r\right)\sum_i\alpha_i.
\end{equation}
\end{proposition}
\begin{proof}
For any clique of $m$ labels,~\eqref{eq:reuse-threshold} gives
\[
 0\le\norm{\sum_{i=1}^{m}t_i}^2
 \le mh^2+m(m-1)(h^2-1/2)
 =\frac m2\bigl(1-m(1-2h^2)\bigr).
\]
This proves the latitude bound. The tails have strictly negative mutual products.
They are affinely independent: an affine dependence would equate
nonempty positive combinations of disjoint sets of tails, whose mutual
inner product is negative, although they are the same vector.
Thus $m\le k+1$ and $\omega(\Gamma)\le r$.
Equation~\eqref{eq:clique-capacity} proves~\eqref{eq:reuse-bound}.
\end{proof}

The multiplicity bound is sharp for one direction at a prescribed
height. For $r\ge2$, a regular simplex of $r$ tails of squared norm
$h^2$ fits in $\R^k$ and has mutual products
$-h^2/(r-1)\le h^2-1/2$; for $r=1$, one tail suffices.
This does not assert simultaneous saturation of block and pure-tail
conditions.

\begin{proposition}[Sharp simplex specialization]\label{prop:simplex}
Let $r\ge2$, let $C\subset S^{D-1}$ be a kissing configuration, and
let $B\subset C$ be nonempty. An equal-latitude lifting replacing
every direction of $B$ by copies at the same $r$ tails, with
$0<h^2<1/2$, exists in $\R^{D+k}$ if and only if
\[
 k\ge r-1,\qquad
 \ip{x}{x'}\le\frac1{r+1}\quad(x\ne x',\ x,x'\in B).
\]
No pure-tail points are required. In particular, there is a kissing
configuration in $\R^{D+r-1}$ with
$|C|+(r-1)|B|$ points. It replaces each direction of $B$ by $r$
copies with tails of squared norm $(r-1)/(2r)$.
\end{proposition}
\begin{proof}
Necessity follows from Proposition~\ref{prop:multiplicity}:
$k\ge r-1$ and $h^2\ge(r-1)/(2r)$. The same-label threshold
$(1/2-h^2)/(1-h^2)$ decreases with $h^2$, so is at most $1/(r+1)$.
For sufficiency, in the hyperplane $\sum_i z_i=0$ of $\R^r$, put
\[
 t_i=\frac1{\sqrt2}\left(e_i-\frac1r\boldsymbol1\right),
 \qquad a=\sqrt{\frac{r+1}{2r}}\qquad(1\le i\le r).
\]
The tails satisfy $\norm{t_i}^2=(r-1)/(2r)$ and
$\ip{t_i}{t_j}=-1/(2r)$ for $i\ne j$. Retain $C\setminus B$
at the equator and use all $(ax,t_i)$ with $x\in B$.
Same-label, distinct-direction products are at most
\[
 \frac{r+1}{2r}\frac1{r+1}+\frac{r-1}{2r}=\frac12.
\]
Distinct labels on the same direction give
$(r+1)/(2r)-1/(2r)=1/2$, and distinct labels on different
directions give at most zero. Theorem~\ref{lem:lift} proves
validity, and~\eqref{eq:lift-count} gives the count.
\end{proof}

At this tail norm the multiplicity bound is $r$, and equality in
its norm-of-sum proof forces the used tails to sum to zero, with
every mutual inner product $-1/(2r)$. Their Gram matrix has rank
$r-1$, so this equality geometry is precisely a regular simplex.
At $h^2=1/4$ the block threshold is $1/3$ and a direction can occur
twice. At $h^2=1/3$ the threshold is $1/4$ and a direction can occur
three times. The Leech constructions realize both cases, fitting
several triangles and pure tails into three tail dimensions in the latter.
The former extends to disjoint blocks as follows.

\begin{theorem}[Antipodal block lift]\label{thm:block-lift}
Let $C\subset S^{D-1}$ be a kissing configuration, and let
$B_1,\ldots,B_t\subset C$ be pairwise disjoint subsets satisfying
\[
 \ip{x}{x'}\le\frac13\qquad(x\ne x',\ x,x'\in B_i).
\]
Let $Y\subset S^{k-1}$ be a kissing configuration containing $t$
distinct antipodal pairs $\{\pm u_i\}$. Then
\begin{align}
 X={}&\{(x,0):x\in C\setminus\textstyle\bigcup_iB_i\}\nonumber\\
 &\quad\cup\{(\tfrac{\sqrt3}{2}x,\tfrac\eps2u_i):
                     x\in B_i,\ \eps\in\{\pm1\}\}
       \cup\{(0,y):y\in Y\}\label{eq:antipodal-lift}
\end{align}
is a kissing configuration in $\R^{D+k}$ with
\begin{equation}\label{eq:antipodal-count}
 |X|=|C|+\sum_i|B_i|+|Y|.
\end{equation}
\end{theorem}
\begin{proof}
Use Theorem~\ref{lem:lift} with tails $\pm u_i/2$. A mixed point has
inner product at most $1/2$ with every pure tail by Cauchy--Schwarz.
For two mixed points in one block at the same tail, the bound is
$(3/4)(1/3)+1/4=1/2$. At opposite tails, different original directions
give at most zero, and equal directions give $3/4-1/4=1/2$.

For points from different blocks, the original directions are distinct,
as are the directed tail labels. Their inner product is at most
$(3/4)(1/2)+(1/4)(1/2)=1/2$. Thus all conditions of the theorem hold.
Each selected direction is deleted once and used twice, giving the count.
\end{proof}

For one block and $Y=\{-1,1\}$ this construction lifts to two latitudes
and includes both poles. In fact their presence fixes the height.
If $(\sqrt{1-h^2}x,\pm h)$ are both used, their mutual inner product
$1-2h^2$ requires $h^2\ge1/4$. The poles require $|h|\le1/2$.
Thus $|h|=1/2$, for $0<|h|<1$. The same-latitude inequality at this
height is exactly the block threshold $1/3$ in the theorem.

Each selected direction adds one point to the orthogonal union of $C$
and $Y$. Section~\ref{app:p48} applies this to large disjoint blocks;
Section~\ref{sec:motion25} varies the two-latitude geometry; and
Section~\ref{sec:leech} uses triple reuse, adding two points per direction.

\section{Coordinated shell motion in dimension 25}\label{sec:motion25}
\providecommand{\ip}[2]{\langle #1,#2\rangle}

\noindent
A coordinated motion of a complete contact layer enlarges the 197579-point
configuration supplied in the public
\href{https://github.com/alexlegeartis/KissingNumbers/tree/cf14c5ef4db6059bbf2e570e3f0ce24edfda243e/verifications/improved/dim25-lens-heads}
{\texttt{KissingNumbers/dim25-lens-heads}} package.
The construction moves 552 retained equatorial points, adjoins one point,
and replaces two cap points. The motion preserves every internal product
of the moved layer.

\begin{lemma}[Coordinated chord motion]\label{lem:chord-motion}
Let $n$ be a unit vector, let $h>0$, and let $r_i\perp n$ satisfy
$\|r_i\|^2=\rho-h^2$. Put $u_i=r_i-hn$ and $w_i=r_i+hn$.
For real $c,d$ with $c^2+d^2=h^2$, set $W_i=(r_i+cn,d)$.
Then $\|W_i\|^2=\rho$ and $W_i\cdot W_j=w_i\cdot w_j$.
If $z\cdot u_i\le\rho/2$ and $z\cdot w_i\le\rho/2$, then
$(z,0)\cdot W_i\le\rho/2$.
\end{lemma}
\begin{proof}
The first two claims follow from $r_i\perp n$ and $c^2+d^2=h^2$.
Since $|c|\le h$, the horizontal component of $W_i$ is
$\frac{h-c}{2h}u_i+\frac{h+c}{2h}w_i$, a convex combination.
The last inequality follows by linearity.
\end{proof}

For an axial candidate $Q=(\alpha n,\beta)$ with
$\alpha^2+\beta^2=\rho$, every moved point has the same product
$\alpha c+\beta d$ with $Q$. Simultaneous contact is possible exactly
when $h^2\ge\rho/4$, at the circle--line intersections
\begin{equation}\label{eq:chord-contact}
 (c,d)=\frac12(\alpha,\beta)
       \mathbin{\pm}\sqrt{\frac{h^2-\rho/4}{\rho}}\,(-\beta,\alpha).
\end{equation}
For a finite nonempty moved layer, compatibility with a fixed point
$(z,s)$ is exactly the half-plane constraint
\begin{equation}\label{eq:chord-boundary}
 c\ip nz+ds\le\rho/2-\max_i\ip{r_i}z.
\end{equation}
The candidate must also be compatible with the fixed points. In the
construction below, two cap replacements resolve the remaining conflicts.

\begin{proposition}[Finite boundary test]\label{prop:chord-finite}
Let $H_1,\ldots,H_m$ be closed half-planes in the $(c,d)$-plane,
and let $F$ consist of $(h,0)$ together with every intersection of
the circle $c^2+d^2=h^2$ with their boundary lines. Constant
inequalities contribute no boundary points. Then $|F|\le2m+1$, and
the circle meets $\bigcap_jH_j$ if and only if some point of $F$
belongs to every $H_j$.
More generally, designate some inequalities as mandatory and assign
nonnegative deletion costs to the others. The minimum cost of violated
inequalities, when the mandatory constraints admit a circle point,
is attained at a feasible point of $F$.
\end{proposition}
\begin{proof}
Remove all boundary intersections from the circle. On each remaining
open arc, every inequality has constant truth value. At an endpoint,
every inequality that held on the arc still holds, because its half-plane
is closed. Moving there preserves mandatory inequalities and cannot
increase deletion cost. If there are no boundary intersections, all truth
values are constant on the circle, and $(h,0)$ suffices.
\end{proof}

For a nonzero normal $A=(A_1,A_2)$, the boundary $A\cdot(c,d)=B$
contributes points exactly when $B^2\le h^2\|A\|^2$, namely
\[
 \frac{BA\pm\sqrt{h^2\|A\|^2-B^2}\,(-A_2,A_1)}{\|A\|^2}.
\]
Rational half-plane data and rational $h^2$ therefore give an exact
test on at most $2m+1$ quadratic-algebraic candidates. This minimizes
deletions; the existence of replacements is a separate requirement.

\paragraph{Normalization and input.}
All points have squared norm 4; the kissing condition is an inner product
at most 2 between distinct points. Write the input configuration as
\[
 C_0=\{(z,0):z\in L\setminus D\}
 \cup\{(x_i,\pm1):0\le i<1016\}
 \cup\{P_0,N,S\},
\]
where $L$ is the 196560-vector minimal shell of the Leech lattice,
$|D|=1016$, $P_0=(P,0)$, and $N=(0,2)$, $S=(0,-2)$.
The head indices refer to rows of the source \texttt{heads\_X.npy}; the
mathematical values of its 44 rational heads are preserved as exact
rational numbers in \texttt{baseline-heads.json}.

Let $v$ be the following squared-norm-six lattice vector associated
with the first complete block, and put $n=v/\sqrt6=-P/2$.
An exact integer description of this direction is
\[
v=\frac{1}{\sqrt8}(-5,-1,-1,-1,-1,1,-1,1,-1,-1,-1,1,
1,-1,-1,-1,1,1,-1,1,-1,-1,1,-1).
\]
The complete block is $U=\{u\in L:\ip uv=-3\}$, of size 552.
Every $u\in U$ belongs to $D$. For each such $u$, the point $w=u+v$ is
minimal, lies outside $D$, and satisfies $\ip uv=-3$ and $\ip uw=1$.
The corresponding old cap points have horizontal coordinate
$u+t v$, with $t=(3-\sqrt3)/6$.
All these input identities are checked by the shell and owner data.

\paragraph{The modification.}
For $u\in U$ define
\[
r=u+\sqrt{3/2}\,n=w-\sqrt{3/2}\,n,\qquad
a=\frac{\sqrt3}{2}+\frac{\sqrt2}{4},\qquad
b=\frac{\sqrt6-2}{4}.
\]
The values $(a,b)$ are the positive choice in~\eqref{eq:chord-contact}
for $h^2=3/2$, $\rho=4$, and $(\alpha,\beta)=(\sqrt3,-1)$.
Then $r\perp n$, $\|r\|^2=5/2$, and the old cap points over $u$ are
$(r-n/\sqrt2,\pm1)$. Replace every $(w,0)$ by
\[
 W(r)=(r+a n,b),
\]
and adjoin
\[
 Q=(\sqrt3 n,-1)=(V/4,-1),
\]
where $V=\sqrt8\,v$ is the integer vector printed above. In particular,
the new point $Q$ is rational and has squared norm $48/16+1=4$.
Finally replace the upper point $(x_{984},1)$ and the lower point
$(x_{1008},-1)$ by $R_+$ and $R_-$, respectively. The certificate specifies
two integer vectors and defines
\[
 R_+=\frac{2k}{\sqrt{1007176}},\qquad
 R_-=\frac{2\ell}{\sqrt{10209}},
\]
where $k,\ell$ are the integer rows in
\texttt{constructions/d25/repair-points.json}, of squared norms
1007176 and 10209.
The opposite old cap points $(x_{984},-1)$ and $(x_{1008},1)$ are retained.

\begin{proposition}[Two necessary cap changes]\label{prop:motion-two-repairs}
Keep $Q$ fixed and move the 552 points by
$W(r)=(r+cn,d)$, with $c^2+d^2=3/2$.
Any compatible configuration obtained this way must change at least
two original cap points. The specified construction attains this minimum.
\end{proposition}
\begin{proof}
The exact cap data give $Q\cdot(x_{1008},-1)>231/100>2$,
so this lower cap must change independently of $(c,d)$.
If it were the only changed cap, all first-block cap pairs
$(r-n/\sqrt2,\pm1)$ would remain. Their products with their own
moved points require $-c/\sqrt2+|d|\le-1/2$.
Together with the circle equation this gives $c>0$ and $|d|\le1/3$.
On this interval $c=\sqrt{3/2-d^2}$ and $\sqrt3c-d$ is strictly
decreasing. Compatibility with $Q$ therefore forces $d\ge b$,
since $\sqrt3a-b=2$.

Put $\nu=\ip n{x_{984}}$ and $M=\max_r\ip r{x_{984}}$.
The largest moved-layer product with the upper 984 cap is
$g(d)=M+\nu\sqrt{3/2-d^2}+d$. Since $\|x_{984}\|=\sqrt3$,
\[
 g'(d)=1-\frac{\nu d}{\sqrt{3/2-d^2}}
 \ge1-\frac{\sqrt6}{5}>0\qquad(b\le d\le1/3).
\]
The rational head and shell data, evaluated with outward-rounded
intervals, give $g(b)>201/100>2$. Thus this upper cap must also
change, a contradiction. The motion-obstruction checker reproduces
both strict checks. The subsequent pair verification proves that
the specified two replacements suffice.
\end{proof}

\paragraph{Pair verification.}
Lemma~\ref{lem:chord-motion} gives the moved-point norms, internal products
and compatibility with the unchanged equator. The identities
\[
a^2+b^2=\frac32,\qquad \sqrt3a-b=2,\qquad
-\frac a{\sqrt2}+b=-\frac34
\]
give $Q^2=4$ and $\ip{W(r)}Q=2$.

The shell has integral $\ip vz\in\{-3,-2,-1,0,1,2,3\}$.
The negative extreme layer is deleted and the positive extreme layer
is exactly the moved set. Thus every unchanged equatorial point obeys
\[
 \ip Q{(z,0)}=\frac{\ip vz}{\sqrt2}\le\sqrt2<2.
\]
The shell checker verifies these layer identities in integers.
Products of $W$ with $(P_0,N,S)$ are $(-2a,2b,-2b)$; those of $Q$
are $(-2\sqrt3,-2,2)$. Against its own first-block caps, $W(r)$ has
products $7/4$ and $11/4-\sqrt6/2<2$; different block indices give
smaller products because $\ip{r_i}{r_j}\le1/2$.

The four remaining groups are: moved points against the 2030 unchanged
caps; $Q$ against those caps; and each of $R_+,R_-$ against all other
final points. Their outward-rounded rational upper bounds are
\[
1.999372,\quad 1.965926,\quad 1.994722,\quad 1.987522.
\]
The checker uses 192-bit Arb intervals~\cite{johansson-arb} and verifies
both repair norms exactly. Unchanged pairs retain the input inequalities.

\begin{theorem}\label{thm:d25}
The specified configuration consists of 197580 distinct points on
the radius-2 sphere in $\mathbb R^{25}$ with pairwise inner products at
most 2. Consequently the kissing number in dimension 25 is at least 197580.
\end{theorem}

\begin{proof}
The pair checks above, together with the integer or rational baseline
checks, exhaust all pairs. All squared norms are exactly 4. The moved
layer and repairs replace equally many old points; only $Q$ increases
the count:
\[
 |C|=197579-552-2+552+2+1=197580.
\]
No duplicate can occur: two equal norm-4 vectors would have inner product
4, contradicting the verified bound 2 for their distinct indices.
\end{proof}

\section{Reuse geometry and two-layer Leech lifting}\label{sec:leech}

The endpoint and four-triangle constructions below isolate the geometry
of direction reuse: the latitude, the tail partition, and the capacity
of a prescribed label system. They give the auxiliary configurations
of sizes 197058 and 200540. Section~\ref{sec:leech-two-layer} then
uses a second layer to obtain the current 201567-point construction.

We use the integer Leech shell $\Sigma\subset\Z^{24}$ in the Golay
coordinate convention of Appendix~\ref{app:leech}. Its relevant
properties are
\begin{equation}\label{eq:leech-shell}
 |\Sigma|=196560,\qquad r\cdot r=32,\qquad
 r\cdot r'\le16\quad(r\ne r',\ r,r'\in\Sigma).
\end{equation}
Consequently $C=\Sigma/\sqrt{32}$ is a kissing configuration. A block
$S\subset\Sigma$ will be called \emph{compatible} if distinct rows
satisfy $r\cdot r'\le8$. After normalization its maximum inner product
is at most $1/4$.

\begin{lemma}[Leech block witnesses]\label{lem:leech-blocks}
There is a compatible set $S_*\subset\Sigma$ of size $496$ and there
are four pairwise disjoint compatible sets $S_0,S_1,S_2,S_3\subset\Sigma$,
each of size $496$.
\end{lemma}
\begin{proof}[Finite verification]
The five complete integer arrays are specified in
Table~\ref{tab:certificates}. They are taken from the
PackingStar coordinate data~\cite{packingstar-data}.
For each array, membership is checked against the three shell families
in Appendix~\ref{app:leech}. Each array has $496$ distinct rows, each
row has squared norm $32$, and the maximum of its $122760$ unordered
distinct-row inner products is eight. The latter four arrays have a
union of cardinality $1984$. These checks, implemented separately by
the two programs described in Section~\ref{sec:certificates}, prove
compatibility and disjointness.
\end{proof}

\subsection{The endpoint construction in dimension 25}
The normalized block $S_*/\sqrt{32}$ satisfies the $1/3$ condition of
Theorem~\ref{thm:block-lift}, since it satisfies the stronger bound $1/4$.
The resulting configuration is
\begin{align}
 X_{25}={}&\{(r/\sqrt{32},0):r\in\Sigma\setminus S_*\}\nonumber\\
 &\quad\cup\{(\sqrt3\,r/\sqrt{128},\eps/2):
                   r\in S_*,\ \eps\in\{\pm1\}\}
       \cup\{(0,1),(0,-1)\}.\label{eq:d25}
\end{align}
Its cardinality is
\[
 196064+992+2=197058,
\]
giving the historical bound $K(25)\ge197058$.
The endpoint choice permits the two poles without requiring a larger
Leech block. Distinct same-latitude directions have inner product at
most $7/16$, whereas the opposite-latitude copies of one direction
have inner product exactly $1/2$. This distinction will also determine
the local geometry of $X_{25}$ in Section~\ref{sec:deformations}.

\subsection{The tail geometry in dimension 27}
Let
\begin{equation}\label{eq:D3}
 V=\{\pm e_i\pm e_j:1\le i<j\le3\},\qquad
 R=\begin{pmatrix}
  1/\sqrt2&-1/\sqrt2&0\\
  1/\sqrt2&1/\sqrt2&0\\
  0&0&1
 \end{pmatrix}.
\end{equation}
The twelve roots in $V$ form a cuboctahedron. Put
\begin{equation}\label{eq:tailsets}
 T=V/\sqrt6,\qquad P=RV/\sqrt2.
\end{equation}
Thus $T$ supplies mixed tails of squared norm $1/3$, and $P$ supplies
twelve unit pure tails. The relative rotation is through $45$ degrees
in the first two coordinates.

\begin{lemma}\label{lem:tails}
Distinct members of $T$ have inner products in
$\{-1/3,-1/6,0,1/6\}$, and $P$ is a kissing configuration. Moreover,
\begin{equation}\label{eq:tail-cross}
 \max_{p\in P,\,t\in T}\ip{p}{t}
  =\frac{1+1/\sqrt2}{\sqrt{12}}<\frac12.
\end{equation}
The roots admit a partition into four zero-sum triples, each of which
has pairwise inner product $-1$.
\end{lemma}
\begin{proof}
Every root has squared norm two. Distinct roots have inner products
$-2,-1,0,$ or $1$, giving the first assertions by scaling and
orthogonality of $R$. The absolute coordinates of $Rv$, up to order,
are either $(\sqrt2,0,0)$ or $(1,1/\sqrt2,1/\sqrt2)$. Its largest inner
product with a root is obtained by choosing two coordinates with the
appropriate signs. The maximum is $1+1/\sqrt2$, giving
\eqref{eq:tail-cross}. The strict inequality follows, for example, from
$18+12\sqrt2<36$ after squaring the equivalent inequality
$\sqrt6+2\sqrt3<6$.

An explicit partition is given by the rows of
\begin{equation}\label{eq:triangles}
\begin{array}{rrr}
 (-1,0,-1)&(1,-1,0)&(0,1,1)\\
 (0,-1,-1)&(-1,0,1)&(1,1,0)\\
 (-1,-1,0)&(0,1,-1)&(1,0,1)\\
 (-1,1,0)&(1,0,-1)&(0,-1,1).
\end{array}
\end{equation}
The roots in a row sum to zero. Their equal squared norms then force
all three mutual inner products to be $-1$.
\end{proof}

Write $T_b$ for row $b$ of~\eqref{eq:triangles}, divided by $\sqrt6$.
Thus $T$ is the disjoint union of $T_0,T_1,T_2,T_3$. A block assigned
to $T_b$ is used at three labels, while its original equatorial copy
is removed once. This is the multiplicity-three counterpart of the
antipodal lift.

\begin{proposition}[Four-triangle lift]\label{prop:triangles}
For pairwise disjoint compatible subsets $S_0,S_1,S_2,S_3$ of $\Sigma$,
the following three sets form a kissing configuration:
\begin{align}
 Z&=\{(r/\sqrt{32},0):r\in\Sigma\setminus\textstyle\bigcup_bS_b\},
 \nonumber\\
 M&=\bigcup_{b=0}^{3}\{(r/\sqrt{48},t):r\in S_b,\ t\in T_b\},
 \label{eq:d27}\\
 P_0&=\{(0,p):p\in P\}.\nonumber
\end{align}
Its cardinality is $196560+2\sum_b|S_b|+12$.
\end{proposition}
\begin{proof}
Use Theorem~\ref{lem:lift} with $h^2=1/3$ and $a=\sqrt{2/3}$.
The mixed vectors have norm one since $32/48+1/3=1$.
For mixed points, multiply their inner-product condition by $48$:
\begin{equation}\label{eq:raw-mixed}
 r\cdot r'+48\ip{t}{t'}\le24.
\end{equation}
If $t=t'$, the raw vectors are distinct and lie in one compatible
block, so the left side is at most $8+16=24$. If $r=r'$ and $t\ne t'$,
disjointness of the blocks places both tails in one triangle, giving
$32-8=24$. If both directions and labels differ, the universal shell
bound and Lemma~\ref{lem:tails} give $16+8=24$ as an upper bound.
These three cases cover every distinct mixed pair. Pure--mixed
compatibility follows from~\eqref{eq:tail-cross}; all remaining pairs
are covered by Theorem~\ref{lem:lift}.

For $B=\sum_b|S_b|$, the layer sizes are $196560-B$, $3B$, and $12$.
They are disjoint, and their parametrizations are injective, as in
\eqref{eq:lift-count}. Their sum is the asserted count.
\end{proof}

Applying Proposition~\ref{prop:triangles} to Lemma~\ref{lem:leech-blocks}
gives
\[
 |Z|=194576,\qquad |M|=5952,\qquad |P_0|=12,
\]
and hence $|X_{27}|=200540$. The comparison allocation uses block
multiplicities $(3,3,2,2,2)$, whereas~\eqref{eq:triangles} uses
$(3,3,3,3)$. Both occupy twelve labels. The latter allocation uses
$496$ fewer original directions and therefore preserves $496$ more
equatorial points.

\subsection{Capacity and its equality case}
The preceding count has an exact interpretation even when the label
sets are not arranged in disjoint blocks. Let $A_i\subseteq\Sigma$
be the rows assigned to tail $t_i\in T$. Assume that the resulting
lifting is valid, retains $\Sigma\setminus\bigcup_iA_i$ at the equator,
and contains all twelve pure tails.

\begin{proposition}\label{prop:capacity}
If $|A_i|\le\alpha$ for each of the twelve labels, then
\begin{equation}\label{eq:capacity}
 |X|\le196560+8\alpha+12.
\end{equation}
Equality holds precisely when every used row occurs at three labels
and every label contains $\alpha$ rows. For $\alpha=496$, the
four-triangle construction attains equality.
\end{proposition}
\begin{proof}
Proposition~\ref{prop:multiplicity} gives maximum multiplicity three.
Writing $m(r)$ for the multiplicity of a used row,
\[
 |X|-196560-12
 =\sum_r(m(r)-1)
 \le\frac23\sum_rm(r)
 =\frac23\sum_i|A_i|\le8\alpha.
\]
The first inequality is an equality exactly when all $m(r)=3$;
the last is an equality exactly when all labels are full. The explicit
assignment has both properties.
\end{proof}

More precisely, let $n_j$ count the used rows with multiplicity $j$,
and let $c_i=\alpha-|A_i|$. The deficit is
\begin{equation}\label{eq:deficit}
 196560+8\alpha+12-|X|
 =\frac23\sum_i c_i+\frac23n_1+\frac13n_2.
\end{equation}
For the comparison allocation, $n_2=3\cdot496$ and all $c_i=0$, so
the deficit is exactly $496$. The capacity assumption in
Proposition~\ref{prop:capacity} concerns the prescribed label system;
the block witnesses supply $\alpha=496$ without requiring a bound on
the largest possible compatible subset of the Leech shell.

\subsection{A two-layer configuration with 201567 points}
\label{sec:leech-two-layer}
The larger construction retains the first-layer heads of the public
201566-point contribution~\cite{lindow27} and replaces its second
layer by a 311-owner witness. We give the pair conditions explicitly.
For this subsection use squared norm 32 and threshold 16 throughout.
Write $\Sigma_{27}$ for the Leech shell generated in the Golay
coordinate basis of the two-layer data. All heads and owners below
use that same basis. The shell argument in Appendix~\ref{app:leech}
applies to its verified Golay generator.

The first-layer data are 2258 distinct integer heads $Y$ of squared
norm 192, each assigned to one of four triples. In the order of the
stored labels, these triples of roots are
\[
\begin{array}{c|rrr}
0&(1,1,0)&(-1,0,-1)&(0,-1,1)\\
1&(1,-1,0)&(-1,0,1)&(0,1,-1)\\
2&(-1,-1,0)&(1,0,-1)&(0,1,1)\\
3&(-1,1,0)&(1,0,1)&(0,-1,-1).
\end{array}
\]
For a head $y$ in triple $b$, take the three points
\begin{equation}\label{eq:d27-first}
 (y/3,\,4v/\sqrt3)\qquad(v\text{ in triple }b).
\end{equation}
The pure-tail points are $(0,4Rv)$ for all twelve roots
$v\in V$, with $R$ as in~\eqref{eq:D3}.

A first-layer owner is a shell vector $z\in\Sigma_{27}$ with $y\cdot z>48$.
The finite data have 2090 heads with exactly one owner and 168
heads with none; all 2090 owners are distinct. The former are
of the form $y=3u+v$ with $u\in\Sigma_{27}$, $v^2=48$, and
$u\cdot v=-24$. The latter are twice norm-48 vectors. These
descriptions explain the candidate families; verification recovers
the owners by checking the whole shell.
For distinct first-layer heads, the stored data satisfy
\begin{equation}\label{eq:d27-first-threshold}
 y\cdot y'\le
 \begin{cases}
 48,&\text{on the same triple},\\
 96,&\text{on different triples}.
 \end{cases}
\end{equation}
Retaining all nonowners at the equator gives
$196560-2090+3(2258)+12=201256$ points.

The additional data are 311 distinct owners $u\in\Sigma_{27}$, with a
label $\ell(u)\in\{1,2,3\}$. None is a first-layer owner. Replace
the equatorial point $(u,0)$ by
\begin{equation}\label{eq:d27-second}
 (\sqrt3\,u/2,\,\eps\sqrt8\,e_{\ell(u)}),
 \qquad \eps\in\{-1,1\}.
\end{equation}
The three labels have respectively 105, 103, and 103 owners.
The checked inequalities are
\begin{align}
 y\cdot u&\le32&&\text{for every first-layer head }y,\label{eq:d27-interlayer}\\
 u\cdot u'&\le8&&\text{if }\ell(u)=\ell(u'),\quad u\ne u',\label{eq:d27-sameline}\\
 u\cdot u'&\le16&&\text{if }\ell(u)\ne\ell(u').\label{eq:d27-crossline}
\end{align}

\begin{theorem}\label{thm:d27}
The two layers just specified, their retained equator, and their
twelve pure-tail points form a kissing configuration of 201567
points in dimension 27.
\end{theorem}
\begin{proof}
All points have squared norm 32. An old equatorial point and a
first-layer cap have product at most $48/3=16$. The two
thresholds in~\eqref{eq:d27-first-threshold} give respectively
$48/9+32/3=16$ and $96/9+16/3=16$.
Two different caps of the same head have product
$192/9-16/3=16$.
The rotation estimate of Lemma~\ref{lem:tails}, scaled by 32,
checks the first layer against all pure tails.

A second-layer cap has product at most
$(\sqrt3/2)16=8\sqrt3<16$ with every retained equatorial
point. Its own owner is deleted. Against a first-layer cap, the
head product and tail product sum to at most
\[
 \frac{\sqrt3}{6}(32)+\frac{4\sqrt8}{\sqrt3}
 =\frac{16\sqrt3+8\sqrt6}{3}<16.
\]
The last inequality is equivalent to
$2\sqrt3+\sqrt6<6$ and follows by squaring positive quantities.
Second-layer caps on distinct owners of the same label have
product at most $(3/4)8+8=14$; distinct labels give at most
$(3/4)16=12$. The two caps of one owner have product
$(3/4)32-8=16$. A second-layer cap and a pure tail have
product at most $\sqrt8\sqrt{32}=16$.
All old--old pairs retain their original bounds.

The head squared norms in the equator, first layer, second layer,
and pure-tail layer are respectively $32$, $64/3$, $24$, and zero,
so different layers cannot coincide. Within a layer, distinct
heads and tail labels give distinct points. Each of the 311
new owners is removed once and inserted twice. Hence
\[
 196560-2090-311+3(2258)+2(311)+12=201567.
\]
Dividing all points by $\sqrt{32}$ proves the unit-sphere claim.
\end{proof}

The finite optimization behind the second layer has a particularly
small geometric description. Keep only shell vectors that are
not first-layer owners and satisfy~\eqref{eq:d27-interlayer}.
Connect two such vectors when their product exceeds eight.
One seeks a largest subset that can be split into three independent
sets of this graph. Assign each independent set to one coordinate
line. Distinct-label products are already controlled by the
Leech-shell bound. The published 310-owner witness and the new
311-owner witness are feasible assignments in this same model;
each extra owner contributes one net point.

\section{Layered binary codes and joint support exchanges}\label{sec:codes}

We now keep the layer geometry fixed and improve its binary ingredients,
using replacements whose cost is the union of their old conflicts.

\subsection{Compatible replacements and shared deletion costs}
Let $X\subset S^{d-1}$ be a kissing configuration and let
$C\subset S^{d-1}\setminus X$ be a finite candidate set. For $c\in C$,
put $F(c)=\{x\in X:\ip{c}{x}>1/2\}$, and join two candidates in a
graph when their inner product exceeds $1/2$. For $I\subseteq C$,
write $F(I)=\bigcup_{c\in I}F(c)$.

\begin{proposition}[Shared-deletion replacement]\label{prop:exchange}
For $D\subseteq X$ and $I\subseteq C$, the set $(X\setminus D)\cup I$
is a kissing configuration if and only if $I$ is independent and
$F(I)\subseteq D$. In particular, for independent $I$, the largest
configuration with inserted set $I$ retains $X\setminus F(I)$ and
has net gain
\begin{equation}\label{eq:general-exchange}
 X_I=(X\setminus F(I))\cup I,\qquad
 |X_I|-|X|=|I|-|F(I)|.
\end{equation}
More generally, for a positive weight function $w$ on $X\cup C$,
the greatest weight gain with inserted set $I$ is
$w(I)-w(F(I))$, where $w(A)=\sum_{a\in A}w(a)$.
\end{proposition}
\begin{proof}
Old--old pairs are compatible; new--new pairs require independence of
$I$; old--new pairs require $F(I)\subseteq D$. Disjointness of $C$ and
$X$ gives the count, and positive weights penalize additional deletions.
\end{proof}

The same count applies to disjoint bundles of points. Partition the
deletable old points into nonempty groups $X_j$, and let $Q_i$ be
pairwise disjoint, internally compatible candidate bundles outside $X$.
Require each $Q_i$ to be compatible with the old points held fixed.
Join two bundles when some cross pair is incompatible, and let $F(i)$
index the old groups containing a conflict with $Q_i$. Put
$a_i=|Q_i|$ and $b_j=|X_j|$. Among replacements that retain or delete
each old group as a whole, the exact optimization problem is
\begin{equation}\label{eq:bundle-exchange}
 \max\left(\sum_i a_i z_i-\sum_j b_j r_j\right),\qquad
 z_i+z_{i'}\le1\ (ii'\in E),\quad
 r_j\ge z_i\ (j\in F(i)),\quad z_i,r_j\in\{0,1\}.
\end{equation}
At an optimum only conflicting groups are removed, each charged once.
Single-point groups recover Proposition~\ref{prop:exchange}; internally
compatible pools admit the following exact matching certificate, an
application of the classical deficiency form of Hall's
theorem~\cite{hall1935} to the bipartite conflict graph.
\begin{proposition}[Optimal subexchange in a compatible pool]
\label{prop:exchange-deficiency}
Suppose the entire candidate set $C$ is internally compatible. Let
$H$ be the bipartite conflict graph on $C$ and $X$, with
$cx\in E(H)$ exactly when $x\in F(c)$, and let $\nu(H)$ be its
maximum matching size. Then
\begin{equation}\label{eq:exchange-deficiency}
 \max_{I\subseteq C}\bigl(|I|-|F(I)|\bigr)=|C|-\nu(H).
\end{equation}
An optimal $I$ is the set of candidate vertices reachable from the
unmatched candidates by alternating paths of a maximum matching.
In particular no positive subexchange exists if and only if a matching
assigns every candidate a distinct old conflict.
The same formula, multiplied by the common bundle size, applies to
internally compatible candidates and old groups of equal cardinality.
\end{proposition}
\begin{proof}
A matching meets at most $|F(I)|$ candidates in $I$, so
$\nu(H)\le |F(I)|+|C\setminus I|$, giving the upper bound in
\eqref{eq:exchange-deficiency}.

From the unmatched candidates of a maximum matching, traverse unmatched
edges toward $X$ and matched edges toward $C$. For the reached sets
$I\subseteq C$ and $J\subseteq X$, every edge out of $I$ reaches $J$
or is the matched edge by which its endpoint in $I$ was reached;
hence $J=F(I)$. No vertex of $J$ is unmatched, since that would
give an augmenting path. Matching edges biject $J$ with the matched
vertices of $I$, leaving exactly $|C|-\nu(H)$ unmatched vertices in
$I$. This proves equality.
\end{proof}

Internal compatibility of the pool is essential: otherwise an
alternating-path set can contain conflicting new candidates, and the
matching deficiency is only an upper bound. The certificate concerns
the specified candidate pool, not all possible replacements.

\subsection{The three-layer construction}
Let $n\ge m\ge32$. Choose a binary code $T\subset\F_2^m$ with minimum
distance at least $\lceil m/4\rceil$, and a support family
$\mathcal B\subseteq\binom{[n]}8$ such that
\begin{equation}\label{eq:support}
 |B\cap B'|\le4\qquad(B\ne B',\ B,B'\in\mathcal B).
\end{equation}
For the standard basis $e_0,\ldots,e_{n-1}$ of $\R^n$, define
\begin{align}
 X_T&=\left\{\frac1{\sqrt m}\sum_{i=0}^{m-1}(-1)^{c_i}e_i:c\in T\right\},
 \label{eq:sign-layer}\\
 X_M&=\left\{\frac1{\sqrt8}\sum_{i\in B}\eps_i e_i:
       B\in\mathcal B,\ \eps_i\in\{\pm1\},\ \prod_{i\in B}\eps_i=1\right\},
 \label{eq:middle-layer}\\
 X_R&=\left\{\frac{\eps e_i+\eps'e_j}{\sqrt2}:
       0\le i<j<n,\ \eps,\eps'\in\{\pm1\}\right\}.
 \label{eq:root-layer}
\end{align}
The final layer consists of the normalized $D_n$ roots.

\begin{proposition}[Three-layer specialization of ERS~\cite{ers}]
\label{prop:ers}
The union $X_T\cup X_M\cup X_R$ is a kissing configuration with
\begin{equation}\label{eq:ers-count}
 |X|=|T|+128|\mathcal B|+2n(n-1).
\end{equation}
\end{proposition}
\begin{proof}
All points have norm one. Support sizes $m$, eight, and two distinguish
the layers, whose sizes are $|T|$, $2^7|\mathcal B|$, and
$4\binom n2$.

In the dense layer, two codewords at distance $d$ have inner product
$1-2d/m\le1/2$. On one middle support, two distinct even sign patterns
differ in at least two positions, giving at most $(8-4)/8=1/2$.
On different middle supports,~\eqref{eq:support} gives at most four
contributions of $1/8$. Two roots with different supports share at most
one coordinate, giving at most $1/2$; on the same support distinct
roots have inner product zero or $-1$.

For the three cross-layer cases, the numbers of overlapping coordinates
give respectively
\begin{equation}\label{eq:ers-cross}
 \ip{X_T}{X_M}\le\sqrt{8/m},\qquad
 \ip{X_T}{X_R}\le\sqrt{2/m},\qquad
 \ip{X_M}{X_R}\le\frac{2}{\sqrt8\sqrt2}=\frac12.
\end{equation}
Here each expression bounds all pairs from the indicated layers,
even when a middle support extends beyond the first $m$ coordinates.
Since $m\ge32$, all bounds are at most $1/2$.
\end{proof}

The inputs are uncoupled: one additional support contributes $128$
points, and each additional sign word contributes one.

\subsection{The support families}
For eight-element supports the identity
\[
 d_H(B,B')=16-2|B\cap B'|
\]
identifies~\eqref{eq:support} with minimum distance eight. The following
finite objects supply the middle layers.

\begin{theorem}\label{thm:supports}
The following constant-weight-code lower bounds hold:
\[
\begin{array}{c|rrrrr}
 n&33&34&37&38&39\\\hline
 A(n,8,8)\ge&1803&1960&2843&3077&3383.
\end{array}
\]
\end{theorem}
\begin{proof}[Finite verification]
In each complete binary-mask list $\mathcal B_n$ of
Table~\ref{tab:certificates}, every mask lies in $[0,2^n)$, has eight
set bits, and occurs once. Pairwise checks give the intersection
distributions in Table~\ref{tab:support-histograms}, all at most four.
Independently, enumerating the five-subsets gives $56|\mathcal B_n|$
distinct keys: a repeated key is equivalent to an intersection of
size at least five. These checks establish the codes.
\end{proof}

\begin{table}[ht]
\centering\small
\caption{Support cardinalities and numbers $h_j$ of unordered pairs
with intersection size $j$. All other intersection counts are zero.}
\label{tab:support-histograms}
\begin{tabular}{rrrrrrr}
\toprule
$n$ & $|\mathcal B_n|$ & $h_0$ & $h_1$ & $h_2$ & $h_3$ & $h_4$\\
\midrule
33 & 1803 & 134282 & 389920 & 680407 & 234225 & 185669\\
34 & 1960 & 168266 & 488949 & 783677 & 281186 & 197742\\
37 & 2843 & 443987 & 1255368 & 1461343 & 621750 & 257455\\
38 & 3077 & 555513 & 1521825 & 1676186 & 706467 & 272435\\
39 & 3383 & 718023 & 1896522 & 1985532 & 820689 & 299887\\
\bottomrule
\end{tabular}

\end{table}

The final support families arise from exchanges in the constant-weight
space. For an incumbent family $\mathcal B$ and a new support $c$, let
$F(c)=\{b\in\mathcal B:|b\cap c|>4\}$. If
$Y\subseteq\binom{[n]}8\setminus\mathcal B$ is internally compatible,
then
\begin{equation}\label{eq:exchange}
 \mathcal B'=(\mathcal B\setminus F(Y))\cup Y,
 \qquad F(Y)=\bigcup_{c\in Y}F(c),
 \qquad |\mathcal B'|-|\mathcal B|=|Y|-|F(Y)|.
\end{equation}
This specializes~\eqref{eq:bundle-exchange} to 128-point bundles;
shared deletions can make a joint exchange profitable when no single
insertion is.

The final transitions are
\[
\begin{array}{c|rrrrr}
n&33&34&37&38&39\\\hline
\text{parent size}&1802&1959&2840&3076&3382\\
|F(Y)|&9&6&37&6&14\\
|Y|&10&7&40&7&15\\
\text{final size}&1803&1960&2843&3077&3383.
\end{array}
\]
The certificates include complete parents, insertions and deletions.
The checker recomputes $F(Y)$ and reconstructs each final family by
\eqref{eq:exchange}; the complete witnesses independently prove the theorem.

The conflict graphs of these five candidate pools have maximum matching
sizes $9,6,37,6,14$, respectively: in each case all removed supports
can be matched to distinct inserted supports. By
Proposition~\ref{prop:exchange-deficiency}, the best subexchange using
only the supplied inserted supports therefore has gain $1,1,3,1,1$.
The displayed exchanges attain these values within the supplied pools,
not necessarily among all replacements.

\subsection{Dense layers from coset unions}
Let $H\subset\F_2^m$ be linear and let $r_0,\ldots,r_{q-1}$ be coset
representatives. Define
\begin{equation}\label{eq:coset-union}
 T=\bigcup_{a=0}^{q-1}(r_a+H),\qquad
 \delta_{ab}=\min_{h\in H}\wt(r_a+r_b+h)\quad(a\ne b).
\end{equation}
Within a coset, distances are the nonzero weights of $H$; between
cosets, the difference set is $r_a+r_b+H$. Thus the $\delta_{ab}$
are their minimum distances, and positivity makes the cosets disjoint,
giving $|T|=q2^{\dim H}$.

\begin{proposition}\label{prop:signs}
The generators and representatives in Appendix~\ref{app:codes} define
the sign codes in Table~\ref{tab:signs}.
\end{proposition}
\begin{table}[ht]
\centering
\caption{Parameters of the dense sign codes.}
\label{tab:signs}
\begin{tabular}{rrrrrrl}
\toprule
$m$&$\dim H$&$d(H)$&$q$&$d(T)$&$|T|$&Used in dimension\\
\midrule
32&17&8&1&8&131072&32--37\\
38&14&12&11&10&180224&38\\
39&15&12&10&10&327680&39\\
\bottomrule
\end{tabular}
\end{table}
\begin{proof}[Finite verification]
Binary elimination gives the three displayed ranks. Enumeration of
the spans gives the weight distributions in Appendix~\ref{app:codes},
whose least positive weights are eight, twelve, and twelve. The
length-$32$ code uses the sole representative zero. For lengths $38$
and $39$, enumeration of~\eqref{eq:coset-union} gives
$\delta_{ab}=10$ for all $55$ and $45$ distinct representative pairs,
respectively. The preceding difference-set argument proves the sizes
and minimum distances. Finally, each distance is at least
$\lceil m/4\rceil$, as required by Proposition~\ref{prop:ers}.
\end{proof}

The length-$32$ code is the Cheng--Sloane
$[32,17,8]$ code~\cite{chengsloane}, in the generator convention of
Grassl's table~\cite{grassl32}. The other two kernels have a common
quadratic-residue origin. Let
\begin{align}
 g(x)={}&x^{23}+x^{19}+x^{18}+x^{14}+x^{13}+x^{12}+x^{10}+x^9
 \nonumber\\
 &\quad+x^7+x^6+x^5+x^3+x^2+x+1\in\F_2[x].\label{eq:qr}
\end{align}
Take the coefficient vectors of $x^ig(x)$, $0\le i<24$, as
length-$47$ binary words, and append their parity at coordinate $47$.
This specifies the extended quadratic-residue code in the chosen
convention~\cite[Chapter~16]{macwilliams}. Shorten at coordinates
$40,\ldots,47$, retaining coordinates $0,\ldots,39$ in order, to
obtain a dimension-$16$ code $H_{40}$. Shortening $H_{40}$ at coordinate
zero gives the printed $H_{39}$; shortening at coordinates zero and
one gives the printed $H_{38}$. Here shortening imposes zeros before
deleting coordinates. Exact row-space comparisons verify these
identities and, in the displayed retained-coordinate conventions,
\begin{equation}\label{eq:nesting}
 H_{38}=\{(h_0,\ldots,h_{37}):(h_0,\ldots,h_{37},0)\in H_{39}\}.
\end{equation}
The distance proof uses the explicit generators and representatives,
not an assumed distance of the mother code.

\begin{proof}[The five layered-code constructions]
For each $n\in\{33,34,37,38,39\}$, combine the support family of
Theorem~\ref{thm:supports} with the sign code of
Proposition~\ref{prop:signs}. Proposition~\ref{prop:ers} gives
Table~\ref{tab:layers}. The dimension-37 and dimension-38 counts are
improved further in Sections~\ref{sec:signed} and~\ref{sec:completion}.
\end{proof}

\begin{table}[ht]
\centering
\caption{The three contributions to the five code-based configurations.}
\label{tab:layers}
\begin{tabular}{rrrrr}
\toprule
$n$&$|X_T|$&$|X_M|$&$|X_R|$&Total\\
\midrule
33&131072&230784&2112&363968\\
34&131072&250880&2244&384196\\
37&131072&363904&2664&497640\\
38&180224&393856&2812&576892\\
39&327680&433024&2964&763668\\
\bottomrule
\end{tabular}
\end{table}

At the first three lengths, the increase relative to the cited
comparison is $128$ times the increase in support size, respectively
$15$, $24$, and $11$. At the last two lengths, both the larger support
families and the explicit coset unions enter the final count.
\subsection{A new length-32 support family}
The same construction applies to a verified family of 1676 eight-subsets
of $[32]$. There are $1,403,650$ unordered support pairs. Their intersection
histogram, in orders zero through four, is
\[
(112077,318935,601033,196122,175483).
\]
No larger intersection occurs. Together with the same $[32,17,8]$ dense
code, Proposition~\ref{prop:ers} gives
\[
K(32)\ge 2^{17}+128(1676)+2(32)(31)=347584.
\]
The family was obtained from Echols's public 1667-word seed~\cite{echols} through joint
exchanges. The final witness, rather than successful replay of a randomized
search, is the finite input to this proof. The 2843-word length-37 family
above is also improved to 2845 words. Section~\ref{sec:signed} uses this
larger family and changes the sign patterns, giving a further gain.

\section{Local signed replacements}\label{sec:signed}
\subsection{A compatible replacement region}
Scale the ERS layers~\cite{ers} of Proposition~\ref{prop:ers} to squared norm 32:
the $[32,17,8]$ code gives $2^{17}$ vectors $(\pm1)^{32}$ padded to
dimension $d$; each eight-subset $B\in\mathcal B$ carries 128
fixed-parity sign patterns with entries $\pm2$; the roots have two
entries $\pm4$. Distinct supports in $\mathcal B$ intersect in at most four.

For $d=35$ and $36$, the archived support families have 2,157 and 2,742
blocks, respectively \cite{brouwer}. The base cardinalities are
\[
 2^{17}+128\cdot2157+4\binom{35}{2}=409548,
 \qquad
 2^{17}+128\cdot2742+4\binom{36}{2}=484568.
\]

\begin{proposition}[Exact bundle boundary]\label{prop:bundle-boundary}
Let $Q\subset\{0,\pm1\}^{d}$ consist of distinct weight-eight vectors
with $q\cdot q'\le4$ for distinct $q,q'$, and let
$\mathcal D\subseteq\mathcal B$.
Deleting the complete middle-layer bundles on $\mathcal D$ and adjoining
$2Q$ preserves the kissing condition and changes the cardinality by
$|Q|-128|\mathcal D|$ if and only if
\begin{equation}\label{eq:bundle-boundary}
 |\supp(q)\cap B|\le4
 \qquad(q\in Q,\ B\in\mathcal B\setminus\mathcal D).
\end{equation}
\end{proposition}
\begin{proof}
Write a retained middle-layer vector as $2s$, and put
$r=|\supp(q)\cap B|$, where $B=\supp(s)$. If $r<8$, signs agreeing
with $q$ on the intersection extend to either prescribed parity on $B$:
one coordinate of $B\setminus\supp(q)$ can fix the parity.
Thus $\max_s q\cdot s=r$. For $r=8$, the maximum is eight or six
according as $q$ has the bundle's parity or not. If $r>4$, a maximizing
$s=q$ violates the additive count, and any other maximizer violates
the kissing condition. This proves necessity, including possible duplicates.

Conversely,~\eqref{eq:bundle-boundary} bounds new--retained middle
products by 16; new--new products have the same bound by hypothesis.
Top--new and bottom--new products are at most $8\cdot2=16$ and
$2\cdot4\cdot2=16$. New vectors have squared norm 32 and distinct
supports from retained middle vectors; coordinate magnitudes distinguish
them from the other layers. Unchanged pairs remain compatible, proving
the count and sufficiency.
\end{proof}

For prescribed $Q$, the least whole-bundle deletion set making $2Q$
a compatible, disjoint insertion is
\[
 F(Q)=\{B\in\mathcal B:\ |B\cap\supp(q)|>4
                         \text{ for some }q\in Q\}.
\]
Its cost is $128|F(Q)|$, as in Proposition~\ref{prop:exchange},
regardless of the new signs' parity while retained bundles remain complete.

\begin{corollary}[Isolated replacement]\label{lem:isolated-replacement}
Let $W$ contain exactly $b$ blocks of $\mathcal B$, and suppose every
other block meets $W$ in at most four coordinates. Any internally
compatible signed weight-eight code $Q$ on $W$ can replace those $b$
bundles, with gain $|Q|-128b$.
\end{corollary}
\begin{proof}
Every new support is contained in $W$, so
\eqref{eq:bundle-boundary} holds after deleting the $b$ contained blocks.
\end{proof}

This is the weight-eight version of the free-zone replacement principle
of Takhanov--Yun~\cite[Definition 5 and Lemma 4]{signed-johnson}.
The exact criterion also permits non-isolated regions with suitable
new supports.

In dimension 35 take $W=\{12,\ldots,22\}$; precisely block 2145 lies inside
$W$. In dimension 36 take $W=\{12,\ldots,23\}$; precisely blocks
2739, 2740 and 2741 lie inside $W$. Indices are zero-based. The independent
checker verifies these facts from the complete support lists, together with
all boundary intersections.

\subsection{Hadamard transfer}
\begin{lemma}[Paired Hadamard transfer]\label{lem:hadamard-transfer}
Let $\mathcal A$ be a family of four-subsets on an even number $m$ of
coordinates, with pairwise intersections at most two. Fix a perfect
matching of the coordinates and retain the blocks meeting each matched
pair in at most one coordinate. If $r$ blocks survive, there exists a
signed weight-eight code of size $16r$ in $\{0,\pm1\}^{m}$ whose distinct
vectors have inner products at most four.
If the paired coordinates are embedded in the ERS configuration, write
$P(A)$ for the union of the four pairs met by a retained source block $A$.
The transferred code replaces any $b$ old bundles with gain $16r-128b$
provided
\begin{equation}\label{eq:transfer-boundary}
 |P(A)\cap B|\le4
 \quad\text{for every retained source block $A$ and undeleted old block $B$}.
\end{equation}
\end{lemma}
\begin{proof}
Take all sixteen sign patterns on each retained block and apply
$H=\left(\begin{smallmatrix}1&1\\1&-1\end{smallmatrix}\right)$ to each
matched coordinate pair. An occupied pair's single entry $\pm1$ becomes
two entries $\pm1$, so every image has weight eight and support $P(A)$.
Original inner products are at most two, by the intersection bound
for different supports and at least one sign change for equal supports.
Since $H^{\mathsf T}H=2I$, products double and distinctness is preserved.
Proposition~\ref{prop:bundle-boundary} gives the replacement claim.
\end{proof}

For dimension 35, an exact-cover construction gives a Steiner quadruple
system $S(3,4,10)$ with 30 blocks. The matching
\[
(0,1),(2,3),(4,5),(6,7),(8,9)
\]
retains 16 blocks, hence 256 signed weight-eight vectors. One extra zero
coordinate embeds them into the eleven-coordinate region $W$.

The twelve-coordinate source is the 51-block design of
Kalbfleisch--Stanton~\cite[Theorem 7]{kalbfleisch-stanton}, also used
in the local-replacement construction of
Takhanov--Yun~\cite[Section V]{signed-johnson}.
For dimension 36, use two copies of a one-factorization of $K_6$. In each
six-coordinate half, take the complements of the three edges in one
factor, giving three internal four-subsets. For each of the five factors,
combine each of its three edges in the first half with each of the three
in the second half. There are $3+3+5\cdot9=51$ blocks, with intersections
at most two. The matching
\[
(0,1),(2,4),(3,5),(6,7),(8,10),(9,11)
\]
retains 36 blocks, hence 576 signed weight-eight vectors.
The same cross-block construction works for every even half-size, and
is sharp in its complete-sign, transverse $2+2$ model. Let $L,R$ be
disjoint sets of $2t$ coordinates, with prescribed perfect matchings
$M_L,M_R$. Write $E_L=\binom L2\setminus M_L$ and
$E_R=\binom R2\setminus M_R$. Identify a transverse $2+2$ support
$e\cup f$ with the pair $(e,f)\in E_L\times E_R$.

\begin{proposition}[Sharp transverse transfer]\label{prop:local36-paper}
Let $t\ge2$ and $\mathcal A\subseteq E_L\times E_R$. Taking all sixteen
signs on each source support and applying the paired Hadamard map gives
an internally compatible weight-eight code if and only if both
\[
 \{f:(e,f)\in\mathcal A\}\quad(e\in E_L),\qquad
 \{e:(e,f)\in\mathcal A\}\quad(f\in E_R)
\]
are matchings on their respective coordinate sets. Consequently
\begin{equation}\label{eq:transverse-capacity}
 |\mathcal A|\le2t^2(t-1),\qquad |Q|\le32t^2(t-1).
\end{equation}
Equality holds precisely when every one of these left and right
fibres is a perfect matching. For any prescribed $M_L,M_R$, equality
is attainable. In particular, $t=3$ gives 576 vectors in twelve coordinates.
\end{proposition}
\begin{proof}
Two source supports with different edges in both halves intersect in
at most two positions. If their left edges agree, their intersection
has size at most two exactly when their right edges are disjoint;
the symmetric statement holds on the right. Complete sign bundles
make this intersection condition necessary as well as sufficient:
signs agreeing on an intersection of size at least three have product
at least three, and hence product at least six after transfer.
Lemma~\ref{lem:hadamard-transfer} proves sufficiency and the factor sixteen.

Each fibre has at most $t$ edges, and
$|E_L|=|E_R|=\binom{2t}{2}-t=2t(t-1)$.
Summing degrees on either side proves the bound, with equality exactly
when every fibre is a perfect matching.

For attainment, take one-factorizations of the two complete graphs;
these exist for every $t$~\cite[p.~137]{kalbfleisch-stanton}. Relabel
vertices so that factor zero is the prescribed pairing in each half,
and use the same $2t-1$ colour labels. For each nonzero colour include
every union of one left edge and one right edge of that colour.
Every fibre is a perfect matching, and there are $(2t-2)t^2$ blocks.
For $t=3$ this is the four-colour construction used above.
\end{proof}

The equality criterion does not classify extremizers as synchronized
one-factorizations. The bound requires both the $2+2$ restriction and
complete source sign bundles; it does not cover arbitrary signed
weight-eight codes.

\subsection{Boundary-constrained capacity}
For fixed old bundles to be deleted, form the bipartite allowed-block
graph $\Gamma\subseteq E_L\times E_R$ by retaining exactly the pairs
$(e,f)$ for which $P(e\cup f)$ satisfies
\eqref{eq:transfer-boundary}. More generally, $\Gamma$ can encode any
prescribed restrictions on the transverse $2+2$ blocks. For each
$e\in E_L$, let $G_e$ be the graph on $R$ with edges
$\{f:(e,f)\in\Gamma\}$; define $H_f$ on $L$ symmetrically. Denote
the matching number of a graph $G$ by $\nu(G)$.

\begin{proposition}[Boundary matching bounds]\label{prop:boundary-capacity}
An allowed family $\mathcal A\subseteq\Gamma$ is compatible exactly
when all its left and right fibres are matchings. If
$C_L\subseteq E_L$, $C_R\subseteq E_R$ form a vertex cover of
$\Gamma$, then
\begin{equation}\label{eq:boundary-cover}
 |\mathcal A|\le
 \sum_{e\in C_L}\nu(G_e)+\sum_{f\in C_R}\nu(H_f).
\end{equation}
In particular one may use either complete vertex class as a cover,
or minimize the displayed weighted cover bound. Any vertex cover of
$G_e$ or $H_f$ gives a further explicit upper bound on its matching number.

The unrestricted bound $2t^2(t-1)$ is attainable under $\Gamma$ if
and only if $\Gamma$ contains a subfamily whose fibres on both sides
are perfect matchings. Thus even one fibre with matching number less
than $t$ obstructs equality. If $U$ is any upper bound obtained from
\eqref{eq:boundary-cover}, deleting $b$ old bundles and using this
model has net gain at most $16U-128b$.
\end{proposition}
\begin{proof}
The compatibility and equality assertions follow from
Proposition~\ref{prop:local36-paper}. The degree of $e$ in
$\mathcal A$ is at most $\nu(G_e)$, and similarly for $f$.
Every selected pair has an endpoint in the chosen vertex cover, so
summing these degree bounds proves~\eqref{eq:boundary-cover}.
A matching uses distinct vertices of any vertex cover, proving the
additional local bound. The last claim counts sixteen inserted points
per source block and 128 points per deleted bundle.
\end{proof}

Local matching bounds need not be simultaneously attainable: choices
in the $G_e$ can violate the $H_f$ conditions. In twelve coordinates,
the geometric boundary resolves this obstruction. Each half contains
three coordinate pairs. A left pair type selects two of the three left
pairs, and a right pair type selects two of the three right pairs.
There are three types on each side. Represent the allowed combinations
by a bipartite graph $H$ on these two three-element sets. Every edge
$a\in E(H)$ allows all sixteen source blocks choosing one coordinate
from each of its four prescribed pairs. This is exactly the restriction
produced by~\eqref{eq:transfer-boundary}, since the image support depends
only on the four occupied pairs, not on the chosen input axes.

Put $h=|E(H)|$. Define a bipartite graph $D_H$ with two copies $a^-,a^+$
of every edge $a$ of $H$, joining $a^-$ to $b^+$ when $a\ne b$ share
an endpoint in $H$. Thus $D_H$ is the bipartite double cover of the
line graph of $H$. Matching--vertex-cover duality bounds the capacity;
compatible source blocks attain the resulting integer cell weights.

\begin{theorem}[Exact twelve-coordinate boundary capacity]
\label{thm:type-boundary-capacity}
The maximum number of compatible transverse $2+2$ source blocks under
the type restriction $H$ is
\begin{equation}\label{eq:type-boundary-capacity}
 \beta(H)=8h-4\nu(D_H)
 =\max_{M\text{ matching in }H}
       \bigl(8|M|+4|E(H-V(M))|\bigr),
\end{equation}
where $V(M)$ is the set of endpoints of $M$. Consequently the maximum
complete-sign transferred code has $16\beta(H)$ vectors.

Fix a maximum matching $N$ in $D_H$. For a compatible allowed family,
let $c_a$ count its source blocks of type $a$. It is optimal if and only if
\[
 c_a+c_b=8\quad(a^-b^+\in N),\qquad
 c_a=8\quad\text{if either copy of $a$ is unmatched by $N$}.
\]
In particular, an optimum can always be chosen with every $c_a$ equal
to zero, four or eight.
\end{theorem}
\begin{proof}
For one type, its four possible left edges each have a right fibre
of size at most two, so $c_a\le8$. If $a,b$ share a left type, their
two distinct right types share a coordinate pair. Every right edge
from these two types meets that pair, so each left fibre has at most
two edges among them. Hence $c_a+c_b\le8$. The symmetric argument
applies when they share a right type.

Apply this last inequality to every edge of $N$, and $c_a\le8$ to
each unmatched vertex copy. Each $c_a$ is counted twice, giving
\[
 2\sum_a c_a\le8|N|+8(2h-2|N|)=16h-8\nu(D_H).
\]
Equality holds exactly when all these inequalities are tight, proving
the asserted equality criterion once attainment is established.

Take a minimum vertex cover $C$ of $D_H$, of size $\nu(D_H)$, and put
\[
 s_a=2-\boldsymbol1_{a^-\in C}-\boldsymbol1_{a^+\in C}\in\{0,1,2\}.
\]
Both directed copies of each adjacency are covered, so
$s_a+s_b\le2$ whenever $a,b$ share an endpoint in $H$.
Thus weight-two types form a matching in $H$ and have only weight-zero
neighbours. To obtain $C$, follow alternating paths from unmatched minus
vertices of $N$, then take unreached minus and reached plus vertices.
Absence of augmenting paths makes this a cover with exactly one endpoint
of each matching edge.

For each type with $s_a=2$, use the eight source blocks whose four
binary coordinate choices have even parity. Distinct choices differ
in at least two positions, so these blocks intersect in at most two
coordinates. For the types with $s_a=1$, restrict one synchronized
36-block factorization family from
Proposition~\ref{prop:local36-paper}. That family has four blocks of
every type: each of its left fibres is a perfect matching on three
coordinate pairs, whose three inter-pair edge counts satisfy
$u+v=u+w=v+w=2$, hence $u=v=w=1$.
The four possible left edges of a type each contribute one block.
Weight-one blocks are mutually compatible. A weight-two type and any
other positive type differ on both sides, so their blocks intersect in
at most one coordinate per half. The combined compatible family has
\[
 4\sum_a s_a=4(2h-|C|)=8h-4\nu(D_H)
\]
blocks, proving attainment.

More generally, for any matching $M$ in $H$, take eight parity blocks
on each edge of $M$ and four factorization blocks on each edge of
$H-V(M)$. The same argument proves compatibility. Conversely, in
the constructed optimum the weight-two types form such an $M$, and
the weight-one types lie in $H-V(M)$. These two observations prove
the second expression in~\eqref{eq:type-boundary-capacity}.
\end{proof}

Thus a fixed twelve-coordinate pairing, a division into two halves,
and $b$ prescribed old-bundle deletions have exact best gain in this model
\[
 16\beta(H)-128b=64\bigl(2h-\nu(D_H)-2b\bigr).
\]
The joint type constraints can be stronger than the fibre-cover bound.
For example, take the five types
$H=\{(1,1),(1,2),(1,3),(2,1),(3,1)\}$.
The two fibre sums are 28 and the best mixed cover gives 24, whereas
$D_H$ has a matching of size five and the exact capacity is 20.
The function \texttt{solve\_type\_boundary} in
\qpath{constructions/codes/hadamard/boundary.py} takes the allowed
types and returns the exact capacity, an attaining source family,
and matching and vertex-cover certificates. The independent tests
check all $2^9$ type relations.
Restrictions on individual source blocks within a type require the
finer boundary model of Proposition~\ref{prop:boundary-capacity}.

\begin{corollary}[One forbidden pair type in twelve coordinates]
\label{cor:forbidden-pair-type}
Let $t=3$, so that each half consists of three prescribed coordinate
pairs. A pair type specifies two of these pairs in each half, and
allows the source edges joining the chosen pairs. If the boundary
forbids at least one such type, at most 32 source blocks, hence 512
signed points, are possible. If it forbids exactly one type and allows
all the others, these bounds are attained.
\end{corollary}
\begin{proof}
There are four left edges joining the two left pairs in the forbidden
type. For each, all permitted right edges meet the third right pair,
a two-vertex cover; its fibre has size at most two instead of three.
The other eight left fibres have size at most three, giving
$4\cdot2+8\cdot3=32$.

The factorization family has exactly four blocks of each pair type,
as in the preceding proof. If only one type is forbidden, deleting
its four blocks leaves 32 compatible, allowed blocks.
\end{proof}

\begin{proposition}[Sign-extension bound]\label{prop:sign-capacity}
Let $Q\subset\{0,\pm1\}^m$ have constant weight $w$ and distinct-pair
inner products at most $s$. If $2w-m>s$, then $|Q|\le2^w$.
In particular, the maximum size of a weight-eight code in eleven
coordinates with inner products at most four is 256.
\end{proposition}
\begin{proof}
A full sign string $\sigma\in\{\pm1\}^m$ cannot extend two members
of $Q$: their supports would intersect in at least $2w-m>s$
coordinates, and their signs agree wherever they intersect.
Each member has $2^{m-w}$ extensions. Hence
$2^{m-w}|Q|\le2^m$. The 256-point transferred code above attains
the bound for $(m,w,s)=(11,8,4)$.
\end{proof}

\begin{theorem}
There are kissing configurations of sizes 409,676 in dimension 35 and
484,760 in dimension 36.
\end{theorem}
\begin{proof}
The two replacement sizes give
\begin{align*}
N_{35}&=2^{17}+128(2157-1)+256+4\binom{35}{2}=409676,\\
N_{36}&=2^{17}+128(2742-3)+576+4\binom{36}{2}=484760.
\end{align*}
Proposition~\ref{prop:ers} controls the unchanged pairs;
Corollary~\ref{lem:isolated-replacement} and
Lemma~\ref{lem:hadamard-transfer} control the replacements.
The archived data verify the support, region and transverse-block hypotheses.
\end{proof}

\subsection{A 512-point replacement in dimension 37}
Proposition~\ref{prop:bundle-boundary} only tests the selected image
supports, rather than their entire containing region. Take the verified
2845-word length-37 family and delete three complete sign bundles.
The replacement consists of 512 signed weight-eight vectors on twelve
coordinates. In the source's one-based, left-to-right numbering, the
paired coordinates are
\[
(18,20),(23,24),(15,21),(16,22),(14,17),(19,25).
\]
Thirty-two transverse four-subsets on the twelve input axes, each with all
sixteen signs, give the 512 vectors by the preceding Hadamard transfer.
The four-subsets, pairs and signed images in
\qpath{constructions/codes/data/d37_signed_patch.json} regenerate exactly
the supplied 512-point set. The checker verifies internal products at
most four and intersections at most four with every retained support.
The boundary-adapted transfer therefore gives
\[
K(37)\ge2^{17}+128(2845-3)+512+2(37)(36)=498024.
\]
Against the 2832-support comparison family, the thirteen extra supports
contribute $13\cdot128=1664$ points. The signed replacement contributes
$512-3\cdot128=128$ further points, for a total gain of 1792.

For a sharp application of Corollary~\ref{cor:forbidden-pair-type},
label the six displayed coordinate pairs $P_0,\ldots,P_5$ in order.
Of the fifteen possible
four-pair image supports, exactly twelve meet every retained old
support in at most four coordinates. The forbidden three are those
whose two omitted pair indices are
\[
 \{0,2\},\qquad\{0,5\},\qquad\{2,4\}.
\]
These three edges form the path $4,2,0,5$. Every division of the six
pairs into two triples cuts at least one edge of this path, because
its four vertices cannot all lie in a single triple. Such a crossing
edge is the omitted pair of a forbidden $2+2$ type, so every division
has capacity at most 512. The division
$\{0,2,4\}\mid\{1,3,5\}$ cuts exactly one edge. Deleting the four
blocks of its forbidden type from a 36-block factorization construction
therefore supplies an alternative boundary-compatible 512-point code.
Consequently 512 is the exact capacity among complete-sign $2+2$
transfers, allowing any division into two triples but keeping these
six coordinate pairs and these three deleted bundles fixed.
More precisely, applying Theorem~\ref{thm:type-boundary-capacity}
to the ten divisions, identifying left--right reversals, gives four
capacities of 512, four of 448 and two of 384. The finite checker
constructs an attaining family for every division and verifies each
transferred code against the complete retained boundary.

The supplied 32-block witness itself fits no such division: its image
supports use ten different four-pair types, whereas a fixed $2+2$
division permits only nine. It belongs to the more general model of
Lemma~\ref{lem:hadamard-transfer}. Neither the restricted capacity
statement nor the 576-point isolated construction in dimension 36
proves optimality for arbitrary twelve-coordinate signed codes.

Left-to-right coordinate $i$ corresponds to integer bit $37-i$ (zero-based),
or coordinate $38-i$ numbered from the right. The verifier uses this
conversion consistently for supports and signed vectors.

\section{Second-shell completion in dimension 38}\label{sec:completion}
The preceding layered-code construction gives 576892 points in
dimension 38. The larger bound in Table~\ref{tab:main} starts from
the public 591612-point configuration~\cite{kissingnumbers}. Its
decomposition has an unoccupied equator. This permits the addition
of directions from a different Leech shell.

\begin{lemma}[Equatorial completion]\label{lem:completion}
Let $q,b>0$. Suppose $X\subset S^{D+k-1}$ is a kissing configuration consisting
of points $(u/\sqrt q,t)$ and pure-tail points $(0,p)$.
Let $V\subset\R^D$ be a set of representatives of distinct
antipodal lines, all of squared norm $b$, such that
\begin{align}
 |\ip{v}{v'}|&\le b/2&&
       (v,v'\in V,\ v\ne v'),\label{eq:completion-new}\\
 |\ip{v}{u}|&\le\sqrt{bq}/2&&
       (v\in V,\ (u/\sqrt q,t)\in X).\label{eq:completion-old}
\end{align}
Assume every $t$ in a point with nonzero head is nonzero.
Then adjoining all $(\pm v/\sqrt b,0)$ gives a kissing configuration
with $|X|+2|V|$ points.
\end{lemma}
\begin{proof}
New points have norm one. Their products with pure tails vanish,
with caps are $\pm\ip vu/\sqrt{bq}$, and with other new lines
are $\pm\ip v{v'}/b$; opposite points have product $-1$.
The nonzero tails separate old and new points.
\end{proof}

\subsection{The base configuration}
Use the integer Golay--Leech shell $\Sigma$ of squared norm 32,
with the Golay generator stored in the dimension-38 package.
Partition its 98280 antipodal lines into classes with pairwise
absolute inner products at most eight. The tail data consist of
1932 vectors of squared norm eight in $\R^{14}$, partitioned into
644 zero-sum triples. Inside a triple the mutual products are $-4$;
between any distinct tail vectors they are at most four.

Each mother line is assigned a triple. It contributes six points
of the form
\begin{equation}\label{eq:d38-base-caps}
 (\eps u/\sqrt{48},\,d/\sqrt{24}),
 \qquad \eps\in\{-1,1\},
\end{equation}
where $d$ runs through its assigned triple.
Their squared norm is $32/48+8/24=1$. Pairs in one class at one
tail have product at most $8/48+8/24=1/2$. Repeated heads
at different tails in their triple have product
$32/48-4/24=1/2$. Different heads and different tails have product
at most $16/48+4/24=1/2$.
Opposite heads give smaller products.

There are also 1932 pure-tail points. Write the tail vectors as the
rows of an integer matrix $F$, and the stored rational rotation as
$M/h$, with $M$ integral and $h>0$. Its defining checks are
\begin{equation}\label{eq:d38-rotation}
 MM^t=h^2I_{14},\qquad
 |(FM^tF^t)_{ij}|^2\le48h^2\quad\text{for all }i,j.
\end{equation}
The pure-tail points are $(0,dM^t/(h\sqrt8))$. Orthogonality preserves
their norms and mutual products. Their products with the caps in
\eqref{eq:d38-base-caps} are entries of
$FM^tF^t/(h\sqrt{192})$, so the second inequality gives absolute
products at most $1/2$. Thus both required checks use exact integer
arithmetic. The layer counts are
\[
 \underbrace{0}_{\text{equator}}+
 \underbrace{6(98280)}_{\text{caps}}+
 \underbrace{1932}_{\text{pure tails}}=591612.
\]
The class partition, triples, rotation, and Golay generator are
inherited data from the public construction; the completion below
is the additional part.

\subsection{The 144 added lines}
\begin{lemma}[Cross-shell compatibility]\label{lem:cross-shell}
Let $\mathcal L$ be a lattice of minimum squared norm $m$.
If $u,v\in\mathcal L$ have distinct squared norms $a,b$, then
$|u\cdot v|\le(a+b-m)/2$.
\end{lemma}
\begin{proof}
The distinct norms imply $u\pm v\ne0$. Applying the minimum bound
to these two lattice vectors gives $a+b\pm2u\cdot v\ge m$.
\end{proof}

\begin{corollary}[An automatic shell window]\label{cor:completion-window}
In Lemma~\ref{lem:completion}, suppose all nonzero heads $u$ are
minimal vectors of one lattice of minimum squared norm $m$, and
$q>m$. For any shell norm $b$ with $m<b\le q$, every lattice vector
$v$ of squared norm $b$ satisfies~\eqref{eq:completion-old}.
Hence any set of lines from that shell satisfying
\eqref{eq:completion-new} gives an equatorial completion.
\end{corollary}
\begin{proof}
Lemma~\ref{lem:cross-shell} gives $|v\cdot u|\le b/2$, so each new--old
product is at most $\frac12\sqrt{b/q}\le1/2$. It is strictly below
$1/2$ when $b<q$. The head norm $\sqrt{m/q}<1$ ensures nonzero cap tails.
\end{proof}

In the integer Leech model, $m=32$ and $q=b=48$ use the endpoint of
this window, giving $|v\cdot u|\le24$ against every minimal vector.
The new certificate lists 144 integer rows $v$ of squared norm
48. They are distinct modulo sign and satisfy
\begin{equation}\label{eq:d38-witness}
 |v\cdot v'|\le24\quad(v\ne v'),\qquad
 \max_{u\in\Sigma}|v\cdot u|\le24.
\end{equation}
Each row belongs to the same Leech lattice: for some
$\varepsilon\in\{0,1\}$ all coordinates have parity $\varepsilon$,
$((v_i-\varepsilon)/2\bmod2)_i$ is a Golay word, and
$\sum_i v_i\equiv4\varepsilon\pmod8$.
Thus only mutual products constrain selection. The checker verifies
membership and recomputes cross products independently. Both arrays
use the base construction's Golay basis.

\begin{proposition}[Contacts at the shell-window endpoint]
\label{prop:shell-endpoint-contacts}
Let $\mathcal L$ have minimum squared norm $m$, and let $v\in\mathcal L$
have squared norm $b$, where $m<b<4m$. At $q=b$, a cap with a minimal
head $u$ is in contact with the equatorial point $v/\sqrt b$ exactly when
$u$ belongs to
\[
 C_v=\{u\in\mathcal L:u^2=(v-u)^2=m\}.
\]
They are paired without fixed points by $u\mapsto v-u$.
For the Leech lattice with $(m,b)=(32,48)$, every such $C_v$ has
552 elements. Each added point in the dimension-38 construction
therefore has exactly 1656 contacts with the old cap layer.
\end{proposition}
\begin{proof}
A minimal head is in contact precisely when $v\cdot u=b/2$,
equivalently $(v-u)^2=m$. The map is an involution; a fixed point
would give $v=2u$ and $b=4m$.

Rescale to minimum four. The Leech minimal shell $L$ is an
11-design~\cite[Theorem 4.1]{nebe-designs}, so for $v^2=s$,
\begin{equation}\label{eq:leech-shell-moments}
 M_{2j}(s):=\sum_{x\in L}(v\cdot x)^{2j}
 =\frac{196560(4s)^j(2j-1)!!}{24\cdot26\cdots(24+2j-2)}
 \quad(1\le j\le5).
\end{equation}
For $s=6$, integrality and the minimum give $|v\cdot x|\le3$.
The polynomial $a^2(a^2-1)(a^2-4)$ isolates the extremes; antipodality gives
\[
 |C_v|=\frac{M_6(6)-5M_4(6)+4M_2(6)}{720}=552.
\]
Each oriented minimal head in~\eqref{eq:d38-base-caps} occurs at
three tails, giving $3(552)=1656$ old cap contacts. Pure-tail points
are orthogonal to the added equatorial point.
\end{proof}

The contact checker verifies the count and involution in integers
for all 144 stored lines.

\begin{theorem}\label{thm:d38}
There is a kissing configuration of 591900 points in dimension 38.
\end{theorem}
\begin{proof}
Apply Lemma~\ref{lem:completion} with $q=b=48$ to the base
configuration~\eqref{eq:d38-base-caps} and the 144-line witness.
Its hypotheses are exactly~\eqref{eq:d38-witness}. The new points
have zero tail, while every old cap has tail norm $1/\sqrt3$;
the old pure-tail points have zero head. Thus the count is
$591612+2(144)=591900$.
\end{proof}

The norm-48 search produced a 142-line witness and then the 144-line
witness used here. All minimal directions occur in the caps, but
they do not exhaust the admissible equatorial directions.

\subsection{Multiple shells and the next-shell obstruction}
The automatic window controls new--old pairs, not pairs between
different added shells. The distinction has an exact formulation.

\begin{proposition}[Multi-shell completion]\label{prop:multi-shell}
Under the hypotheses of Corollary~\ref{cor:completion-window}, let
$V_i\subset\mathcal L$ represent distinct antipodal lines of squared norm $b_i$, where
$m<b_i\le q$. Adjoining $(\pm v/\sqrt{b_i},0)$ for every $v\in V_i$
gives $|X|+2\sum_i|V_i|$ distinct kissing points if and only if
\begin{equation}\label{eq:multi-shell-separation}
 |v\cdot w|\le\frac{\sqrt{b_i b_j}}2
 \quad\Longleftrightarrow\quad
 \min\{\|v-w\|^2,\|v+w\|^2\}\ge b_i+b_j-\sqrt{b_i b_j}
\end{equation}
for every two distinct labelled representatives $(v,i),(w,j)$.
\end{proposition}
\begin{proof}
All new--old products satisfy the automatic shell bound. Products
between the two signs of one new line equal $-1$, while the largest
product between two labelled lines is $|v\cdot w|/\sqrt{b_i b_j}$.
This proves necessity, sufficiency and distinctness. The equivalence
follows from $\|v\pm w\|^2=b_i+b_j\pm2v\cdot w$.
\end{proof}

The ambient minimum bound is weaker than the required separation:
for $b,c>m$,
\[
 b+c-\sqrt{bc}\ge\sqrt{bc}>m.
\]
In the current Leech normalization, evenness makes every squared norm
a multiple of 16. The automatic interval $32<b\le48$ therefore contains
only the norm-48 shell. The next shell is excluded more strongly:

\begin{proposition}[No norm-64 equatorial addition]\label{prop:no-norm64}
Every norm-64 Leech vector has exactly 46 minimal vectors $u$ with
$v\cdot u=32$, and 46 with $v\cdot u=-32$.
Neither $v/8$ nor $-v/8$ can be adjoined at the equator of the
unchanged base configuration~\eqref{eq:d38-base-caps}.
\end{proposition}
\begin{proof}
Rescale to minimum four and $v^2=8$. For $a_x=v\cdot x$ on the
minimal shell, integrality and the minimum give $|a_x|\le4$.
Using~\eqref{eq:leech-shell-moments} and the extreme-value polynomial
$P(a)=a^2(a^2-1)(a^2-4)(a^2-9)$ gives
\[
 |\{x:a_x=4\}|=
 \frac{M_8(8)-14M_6(8)+49M_4(8)-36M_2(8)}{40320}=46.
\]
The negative count is the same. Returning to integer coordinates,
each of these positive heads gives the forbidden product
$32/\sqrt{64\cdot48}=1/\sqrt3>1/2$ with $v/8$.
Only this extreme positive product violates the bound. Every oriented
minimal head appears at three tails, so each proposed endpoint has
exactly 138 conflicts with the old cap layer.
\end{proof}

\section{Seven lifts from a compatible line class}\label{app:p48}

One compatible $P_{48p}$ line class gives seven antipodal block lifts.
Their gains depend on both the class size and the union of its selected images.

\subsection{The tail configurations}
For dimensions $k=1,\ldots,7$, use respectively the normalized root
configurations
$A_1,A_2,D_3,D_4,D_5,E_6,E_7$. They have
$2,6,12,24,40,72,126$ points and are antipodal. Explicitly,
$A_1=\{-1,1\}$, and $A_2$ is the six points
$(e_i-e_j)/\sqrt2$ in the plane $x_1+x_2+x_3=0$. The $D_k$ roots
are $(\eps e_i+\eps'e_j)/\sqrt2$, with $i<j$. Their kissing
conditions follow directly from their coordinate overlaps.

The last two codes come from the integer $E_8$ model
\begin{equation}\label{eq:E8}
 W=\{\pm2e_i\pm2e_j:i<j\}
  \cup\{(\eps_0,\ldots,\eps_7):\eps_i\in\{\pm1\},\ \prod_i\eps_i=1\}.
\end{equation}
Coordinate overlaps and even sign distance give squared norm eight
and distinct-row products at most four. The subsets
\[
 W_7=\{w\in W:w_6+w_7=0\},\qquad
 W_6=\{w\in W:w_5=w_6=-w_7\}
\]
have respectively $62+64=126$ and $40+32=72$ rows. The maps
\[
 w\longmapsto(w_0,\ldots,w_5,\sqrt2w_6),\qquad
 w\longmapsto(w_0,\ldots,w_4,\sqrt3w_5)
\]
are isometries on the indicated subspaces. Dividing by $\sqrt8$
therefore gives the required $7$- and $6$-dimensional antipodal codes.

\subsection{Rigidity of antipodal lifts with pure tails}
Pure tails fix the lifting geometry even when paired heights may differ
and depend on the head.

\begin{proposition}[Paired-lift rigidity]\label{prop:paired-rigidity}
Let $C\subset S^{D-1}$ and $Y\subset S^{k-1}$ be kissing
configurations, with $Y$ containing distinct antipodal pairs
$\{\pm u_i\}$, $1\le i\le t$. Choose $B_i\subseteq C$ and, for
each $x\in B_i$, heights $0<h_{i,x}^+,h_{i,x}^-<1$. Retain
$C\setminus\bigcup_iB_i$ at the equator, retain all of $Y$ as pure
tails, and insert
\[
 \left(\sqrt{1-(h_{i,x}^{\pm})^2}\,x,
                         \ \pm h_{i,x}^{\pm}u_i\right).
\]
The resulting set is a kissing configuration if and only if every
height is $1/2$, the sets $B_i$ are pairwise disjoint, and distinct
heads in each $B_i$ have inner product at most $1/3$.
\end{proposition}
\begin{proof}
The pure tail of the same sign requires $h_{i,x}^{\pm}\le1/2$.
Write these two heights as $\sin\theta_+$ and $\sin\theta_-$,
where $0<\theta_{\pm}\le\pi/6$. The inner product of the two
copies of $x$ is $\cos(\theta_++\theta_-)$. It is at most $1/2$
only if $\theta_++\theta_-\ge\pi/3$, forcing both heights to
equal $1/2$.

If the same head belongs to two blocks, choose the signs of its tails
to make their product nonnegative. The two corresponding points have
inner product $3/4+|\langle u_i,u_j\rangle|/4\ge3/4$, a
contradiction. Within one block, the same-tail condition is exactly
$(3/4)\langle x,x'\rangle+1/4\le1/2$. These prove necessity;
Theorem~\ref{thm:block-lift} proves sufficiency.
\end{proof}

Now suppose $C$ is antipodal. For a pool
of compatible antipodal line classes $\mathcal B_1,\ldots,\mathcal B_m$ of $C$,
suppose each of $t\le m$ tail pairs is assigned a subset of one selected
class. The largest configuration obtainable in this model has exactly
\begin{equation}\label{eq:pool-optimum}
 |C|+|Y|+2\max_{|I|=t}\left|\bigcup_{i\in I}\mathcal B_i\right|
\end{equation}
points: disjointness gives the upper bound, and assigning each line,
with both signs, to its first selected class attains it.

\subsection{An orbitwise image-union guarantee}
Standard averaging and conditional expectations~\cite{raghavan1988}
give the following bound without assuming transitivity.

\begin{proposition}[Orbitwise image union]\label{prop:orbit-union}
Let a finite group $\mathcal G$ act on a finite set $\Omega$, with
orbits $O_j$ of sizes $N_j$. For $\mathcal B\subseteq\Omega$, put
$b_j=|\mathcal B\cap O_j|$ and $q_j=1-b_j/N_j$. For every integer
$t\ge1$, there are $g_1,\ldots,g_t\in\mathcal G$ such that
\begin{equation}\label{eq:orbit-union}
 \left|\bigcup_{i=1}^t\mathcal B g_i\right|
 \ge\left\lceil\sum_jN_j(1-q_j^t)\right\rceil.
\end{equation}
They can be selected deterministically: if $U_r$ is the union after
$r<t$ choices, choose $g_{r+1}\in\mathcal G$ to maximize
\begin{equation}\label{eq:orbit-selection}
 \sum_jq_j^{\,t-r-1}
       |(\mathcal B g_{r+1}\setminus U_r)\cap O_j|.
\end{equation}
Here $q_j^0=1$, including when $q_j=0$.
\end{proposition}
\begin{proof}
Each point of $O_j$ is missed by a uniform image with probability $q_j$,
and by $t$ independent images with probability $q_j^t$. After $r$ choices
the conditional expected union size is
\[
 \Phi_r=\sum_j\bigl[N_j-q_j^{\,t-r}(N_j-|U_r\cap O_j|)\bigr].
\]
Rule~\eqref{eq:orbit-selection} maximizes $\Phi_{r+1}$, whose average
over the next image is $\Phi_r$. Hence $|U_t|=\Phi_t\ge\Phi_0$;
integrality gives the ceiling.
\end{proof}

For compatible line classes, each covered line contributes two points
to the lift. A transitive action gives the expected union
$N[1-(1-b/N)^t]$ for $N$ lines and a $b$-line class. The seven bounds
below use certified image unions, not this average.

\subsection{Selection from a finite image pool}
For a directly computable guarantee, let $B_1,\ldots,B_m$ be a finite
pool of supplied line classes, put $\Omega=\bigcup_iB_i$, and let
$a_x=|\{i:x\in B_i\}|$. Distinct indices are selected without
replacement; the classes themselves may overlap.

\begin{proposition}[Finite-pool selection]\label{prop:finite-pool}
For $0\le t\le m$, some $t$ classes have union size at least
\begin{equation}\label{eq:finite-pool-bound}
 \left\lceil\sum_{x\in\Omega}
 \left(1-\frac{\binom{m-a_x}{t}}{\binom mt}\right)\right\rceil.
\end{equation}
Such a selection can be obtained using integer scores. At a stage with
$M$ classes remaining and $h$ choices still to make, let $U$ be the
current union and $a_x$ the number of remaining classes containing
$x\notin U$. Choose a remaining class $B_i$ maximizing
\begin{equation}\label{eq:finite-pool-score}
 \sum_{x\in B_i\setminus U}\binom{M-1-a_x}{h-1}.
\end{equation}
Here a binomial coefficient is zero if its upper index is negative
or smaller than its lower index.
\end{proposition}
\begin{proof}
A uniform $t$-subset of the $m$ indices misses $x$ with probability
$\binom{m-a_x}{t}/\binom mt$, giving~\eqref{eq:finite-pool-bound}
by averaging. Conditional on the choices already made, the expected
final size after choosing $i$ is
\[
 |\Omega|-\frac{1}{\binom{M-1}{h-1}}
       \sum_{x\notin U\cup B_i}\binom{M-1-a_x}{h-1}.
\]
Thus maximizing~\eqref{eq:finite-pool-score} maximizes this
conditional expectation. Points in every remaining class contribute
equally for every choice and are covered at this step; the stated
zero convention handles them. The average over the next choice is
the current expectation, so it never decreases. At the final step
it equals the actual union size, proving the bound.
\end{proof}

The implementation in \texttt{constructions/p48/coverage.py} collects
equal incidence sets $\{i:x\in B_i\}$ with their multiplicities.
Its subsequent exchanges use exact union gains, preserving the guarantee.

\subsection{The 7077-line class and its automorphism images}
\label{sec:p48-update}
We use the catalogue lattice $P_{48p}$ in its published
integral Gram basis $G$~\cite{p48p-catalogue}. Its minimum squared norm
is six. Since it is even unimodular of rank 48, the theta-series
calculation in~\eqref{eq:theta48} gives 52416000 minimal vectors.
The witnesses use this basis directly, not the quadratic-residue
neighbour of Section~\ref{sec:qr-neighbour}.

Represent a minimal line by an integer row $c$ satisfying
$cGc^t=6$, identifying $c$ and $-c$. The new mother class
$\mathcal C$ has 7077 lines and satisfies
\begin{equation}\label{eq:p48-class}
 |cGc'^t|\le2\qquad(c\ne\pm c',\ c,c'\in\mathcal C).
\end{equation}
Every norm and pair product in this assertion is an integer.
The certificate supplies the complete class and matrices $A_i$ with
$A_iGA_i^t=G$. Thus $\mathcal C A_i$ is another class with the
same property.

For $t$ selected images, assign each line to the first image in which
it occurs. Write $\mathcal D_i$ for the assigned subset of
$\mathcal C A_i$, and put
\[
 S=\left|\bigcup_{i=1}^t\mathcal C A_i\right|
   =\sum_{i=1}^t|\mathcal D_i|.
\]
The $\mathcal D_i$ are disjoint as line sets. Taking both signs of
each assigned line gives disjoint vector blocks $B_i$ of size
$2|\mathcal D_i|$, with normalized inner products at most $1/3$
by~\eqref{eq:p48-class}. Assign these blocks injectively to the $t$
distinct antipodal lines of the corresponding root tail code.
Theorem~\ref{thm:block-lift} gives
\begin{equation}\label{eq:p48-lines-count}
 K(48+k)\ge52416000+2S+\tau_k,
 \qquad \tau_k=2t.
\end{equation}

\begin{theorem}\label{thm:p48-update}
The class and image matrices in the supplementary data give the
seven bounds in Table~\ref{tab:p48-update}.
\end{theorem}
\begin{proof}
The finite checker verifies the catalogue Gram matrix, its exact
isometries, all 7077 norms, distinctness modulo sign, and
all distinct-line inequalities~\eqref{eq:p48-class}. It then
reconstructs each selected image and counts its union with the
preceding images. This gives the $S$ values in the table.
The root tail codes have the ranks, sizes, and inner products
established above. Hence~\eqref{eq:p48-lines-count} applies.
\end{proof}

\begin{table}[ht]
\centering
\caption{The line-class liftings. The union, rather than
the sum of image sizes, determines the number of available lines.}
\label{tab:p48-update}
\begin{tabular}{rrrrr}
\toprule
$d$&$t$&$S$&$\tau_k$&$52416000+2S+\tau_k$\\\midrule
49&1&7077&2&52430156\\
50&3&21231&6&52458468\\
51&6&42459&12&52500930\\
52&12&84874&24&52585772\\
53&20&141328&40&52698696\\
54&36&253851&72&52923774\\
55&63&442507&126&53301140\\\bottomrule
\end{tabular}
\end{table}

\subsection{Class growth and image selection}
The public comparison uses a 7069-line class
$\mathcal C_0\subset\mathcal C$~\cite{kissingnumbers}. With the selected
matrices held fixed, let
\[
 S_0=\left|\bigcup_{i=1}^t\mathcal C_0 A_i\right|,
\]
and let $S_{\rm ref}$ be the union size in the comparison construction.
The increase in the final point count splits exactly as
\begin{equation}\label{eq:p48-gain}
 2(S-S_{\rm ref})=2(S-S_0)+2(S_0-S_{\rm ref}).
\end{equation}
The terms measure class enlargement with images fixed and image selection
with the original class fixed, respectively.

\begin{table}[ht]
\centering
\caption{The two contributions to the seven point-count improvements.}
\label{tab:p48-gain}
\begin{tabular}{rrrrr}
\toprule
$d$&$S_0$&$2(S-S_0)$&$2(S_0-S_{\rm ref})$&Total\\\midrule
49&7069&16&0&16\\
50&21207&48&2&50\\
51&42411&96&18&114\\
52&84778&192&64&256\\
53&141168&320&154&474\\
54&253565&572&296&868\\
55&442009&996&414&1410\\\bottomrule
\end{tabular}
\end{table}

The three selected images in dimension 50 are disjoint. Larger image
families overlap, so an eight-line enlargement need not supply eight
new lines per image. In dimension 55 it contributes 498 distinct lines,
rather than $8\cdot63=504$. The image choice contributes a further
207 lines.

\section{Embedded root systems and moment-certified sections}\label{sec:sections}
An embedded root system imposes restrictions on the projections of the
entire minimal shell. In the case below, these restrictions determine
the number of points perpendicular to a child section.

\begin{theorem}[$E_8$-extension count]\label{d43:thm:structure}
Let $L$ be an extremal even unimodular lattice of rank 48, and let
$\iota:\sqrt3E_8\hookrightarrow L$ be an isometric embedding.
Use the simple-root numbering in~\eqref{d43:eq:simple}, and set
\[
 U=\operatorname{span}_{\R}\{\iota(\sqrt3\alpha_1),\ldots,
                              \iota(\sqrt3\alpha_5)\}.
\]
Then exactly $2,553,792$ vectors of squared norm 6 in $L$ are
perpendicular to $U$. The same count holds for every $W(E_8)$-conjugate
of this $D_5$ subsystem inside the embedded $E_8$.
\end{theorem}

The proof uses integrality, the lattice minimum and spherical-design
moments, not an extension of the Weyl group to automorphisms of $L$.

\subsection{Signature domains and target statistics}
The standard design identities~\cite{dgs,nebe-designs} give the following
principle. The signature domain must be complete, and averaging must
preserve the target statistic.

\begin{proposition}[Finite-section statistics]\label{prop:section-statistics}
Let $L\subset\R^n$ be an integral lattice of minimum squared norm $\mu$,
and suppose its minimal shell $X$ has $N$ points and is a spherical
$(2s+1)$-design. Choose independent $b_1,\ldots,b_r\in L$ with Gram
matrix $H$. For $x\in X$ set
\[
y(x)=(\ip{x}{b_i})_{i=1}^r,\qquad
q(y)=y^tH^{-1}y,\qquad f(y)=|\{x\in X:y(x)=y\}|.
\]
Choose a finite antipodal set $\mathcal R\subset\Z^r\setminus\{0\}$
of section-vector coefficients, and put
\begin{equation}\label{eq:section-domain}
\begin{split}
E&=\{Hz:z\in\mathcal R,\ z^tHz=\mu\},\\
D&=\{y\in\Z^r:q(y)\le\mu,\quad
             |z^ty|\le z^tHz/2\ \text{for every }z\in\mathcal R\}.
\end{split}
\end{equation}
Then every signature lies in the disjoint union $D\cup E$, and
$f(e)=1$ for each $e\in E$. No primitivity of the section is required.

For a multi-index $\alpha$ of total degree at most $2s+1$, define
$M_\alpha=0$ when $|\alpha|$ is odd, and, when $|\alpha|=2m$, define
\begin{equation}\label{eq:section-moments}
M_\alpha=
\frac{N\mu^m}{\prod_{j=0}^{m-1}(n+2j)}
\sum_P\prod_{(a,b)\in P}H_{i_a i_b},
\end{equation}
where the labelled list $(i_1,\ldots,i_{2m})$ contains $\alpha_i$
copies of $i$, and $P$ runs over its pairings. Empty products give
$M_0=N$.

Let $\Gamma$ be a finite group of linear transformations preserving
$q$, $D$ and $E$, and let $\psi:D\cup E\to\R$ be $\Gamma$-invariant.
For its orbits $O_j$ on $D$, form
\[
A_{\alpha j}=\sum_{y\in O_j}y^\alpha,\qquad
b_\alpha=M_\alpha-\sum_{e\in E}e^\alpha,\qquad
c_j=\sum_{y\in O_j}\psi(y).
\]
Any selected moment rows in the stated degree range may be used. If
$A^t\lambda\le c$ entrywise, then
\begin{equation}\label{eq:section-statistic-bound}
\sum_{x\in X}\psi(y(x))\ \ge\
\sum_{e\in E}\psi(e)+\lambda^tb.
\end{equation}
If $A^t\lambda=c$, equality holds.
\end{proposition}
\begin{proof}
Integrality and orthogonal projection give $y(x)\in\Z^r$ and
$q(y(x))\le\mu$. For $a=\sum_i z_i b_i$ and $x\ne\pm a$,
$\|x\pm a\|^2\ge\mu$ gives $|z^ty(x)|\le\|a\|^2/2$.
The exceptions have $\|a\|^2=\mu$: their full norm lies in the
section, so each fibre is a singleton and violates its own domain
inequality. Thus $D\cap E=\varnothing$; positive definiteness of $q$
makes $D$ finite.

Polarization gives~\eqref{eq:section-moments}. Since every isometry of
$q$ extends orthogonally to $\R^n$, these moments are $\Gamma$-invariant
even when $\Gamma$ does not act on $L$. Hence
\[
\bar f(y)=\frac1{|\Gamma|}\sum_{g\in\Gamma}f(g^{-1}y).
\]
is nonnegative, has the same moments and equals one on $E$.
Its orbit values satisfy $Av=b$, while invariance of $\psi$ preserves
$\sum_X\psi=\sum_E\psi+c^tv$. The result follows from
$(c-A^t\lambda)^tv\ge0$.
\end{proof}

Antipodality removes odd total-degree rows, not monomials with odd
individual exponents. Exceptional orbits may instead be retained with
value one. A noninvariant target is only averaged, not determined for a
fixed child. In dimension 45, longer section vectors remove all columns
carrying slack in a one-sided certificate.

\subsection{A realization in \texorpdfstring{$P_{48n}$}{P48n}}
Take $L=P_{48n}$, whose Gram matrix $G$ is published as CQ48a
in the Nebe--Sloane catalogue~\cite{latticecatalogue}. It is an extremal even
unimodular lattice of rank 48, minimum squared norm 6, with minimal shell
\[
X=\{x\in L:\|x\|^2=6\},\qquad |X|=52,416,000.
\]
The parent certificate~\cite{crosssections} specifies eight integer
coefficient rows $T$ and scaled coordinate rows $B$,
satisfying
\begin{equation}\label{d43:eq:embedding}
 TGT^t=H=3C(E_8)=BB^t/96.
\end{equation}
The catalogue and certificate Grams agree entry by entry. The first five
rows $T_5$ have Gram
\begin{equation}\label{d43:eq:d5}
H_5=\begin{pmatrix}
6&0&-3&0&0\\0&6&0&-3&0\\-3&0&6&-3&0\\
0&-3&-3&6&-3\\0&0&0&-3&6
\end{pmatrix},\qquad \det H_5=972.
\end{equation}
This is $3C(D_5)$ in the stated numbering. Both embeddings are primitive
(all nonzero Smith invariants are one), though the theorem does not
require this. For $U=\operatorname{span}_{\R}(T_5)$, $\dim U^\perp=43$.

\begin{corollary}\label{d43:thm:main}
The set $Y=X\cap U^\perp$ has $2,553,792$ elements. Consequently the ordinary
Euclidean kissing number satisfies $K(43)\ge2,553,792$.
\end{corollary}

Theorem~\ref{d43:thm:structure} gives the cardinality. Since
$\|x-x'\|^2\ge6$ for distinct shell vectors, $Y/\sqrt6$ is a
kissing configuration.

The construction is specified without a list of millions of points. In
integer basis coordinates its defining set is
\begin{equation}\label{d43:eq:integer}
\mathcal Z=\{z\in\Z^{48}:zGz^t=6,\ T_5Gz^t=0\}.
\end{equation}
For $EE^t=G$, its points are $zE/\sqrt6$ with $z\in\mathcal Z$;
the certificate supplies $G,T_5$.

\subsection{The complete projection domain}
Now let $L$ be any lattice satisfying Theorem~\ref{d43:thm:structure},
and let $X$ be its squared-norm-six shell. The theta-series
identity~\eqref{eq:theta48} gives $|X|=52,416,000$.
Identify the parent section with $\sqrt3E_8$, with simple roots
$\sqrt3\alpha_1,\ldots,\sqrt3\alpha_8$. Write the orthogonal projection of
$x\in X$ onto this section as $p=\lambda/\sqrt3$.
Integrality of $L$ and self-duality of $E_8$ imply $\lambda\in E_8$. Set
\[
y_i=\ip{\lambda}{\alpha_i}=\ip{x}{\sqrt3\alpha_i},\qquad
q=\|p\|^2=\|\lambda\|^2/3=y^tH^{-1}y\le6.
\]
For an $E_8$ root $r$, if $x\ne\pm\sqrt3r$, applying the minimum bound to
$x\pm\sqrt3r$ gives
\begin{equation}\label{d43:eq:rootbound}
|\ip{\lambda}{r}|\le3.
\end{equation}
The exceptions are $\lambda=3r$. Each such projection has multiplicity one:
its squared length is already 6, so the perpendicular component vanishes.

For a transparent finite enumeration use the model
\[
E_8=\{u/2:u\in\Z^8,\ u_i\text{ all of the same parity},\
\textstyle\sum_i u_i\equiv0\pmod4\}.
\]
Its roots are the 112 vectors $\pm e_i\pm e_j$ and the 128 half-integer
vectors with entries $\pm1/2$ and an even number of minus signs.
The simple roots used here are
\begin{equation}\label{d43:eq:simple}
\begin{aligned}
\alpha_1&=(1,-1,-1,-1,-1,-1,-1,1)/2, &\alpha_2&=e_1+e_2,\\
\alpha_{i+2}&=-e_i+e_{i+1} &&(1\le i\le6).
\end{aligned}
\end{equation}
Their Gram is precisely $H/3$.

Enumeration of $\|u\|^2\le72$ subject to~\eqref{d43:eq:rootbound} gives
\begin{center}
\begin{tabular}{rrrrrr}\toprule
$\|\lambda\|^2$ &0&2&4&6&8\\\midrule
Nonexceptional signatures &1&240&2,160&6,720&17,280\\\bottomrule
\end{tabular}
\end{center}
There are no other nonexceptional signatures. Including the 240 exceptions
gives 26,641 signatures in total.

For completeness, pair roots bound the two largest absolute coordinates
of $u$ by a sum of 6. Half-integer roots bound $\|u\|_1$ by 12,
except when all coordinates are nonzero with odd sign parity, when they
bound $\|u\|_1-2\min_i|u_i|$ by 12. Thus $\|u\|_1\le18$ is a
safe enumeration bound. Exact parity, norm and root tests give the table;
at norm eight they remove the 240 doubled roots from 17,520 points.

\subsection{Moments and the averaging explanation}
Venkov's theorem~\cite[Theorem 4.1]{nebe-designs} makes $X$ an 11-design,
so~\eqref{eq:section-moments} applies through degree ten. The abstract
Weyl group has five nonexceptional orbits, with sizes $s_j$ and projected
squared norms $q_j=0,2/3,4/3,2,8/3$. Their averaged multiplicities $a_j$
satisfy the radial equations
\begin{equation}\label{d43:eq:radial}
\sum_j s_ja_jq_j^k=|X|6^k\frac{(4)_k}{(24)_k}-240\,6^k,
\qquad 0\le k\le5,
\end{equation}
where $(a)_k$ is the rising factorial and the subtraction removes the
240 section roots. The first five equations give
\[
(a_0,a_1,a_2,a_3,a_4)
=(1,092,000,116,235,9,180,495,63/4).
\]
The equation for $k=5$ also holds.

The roots perpendicular to the first five simple roots form $A_3$.
The target subspace contains respectively $1,12,6,24,0$ admissible
signatures at the five norms, and 12 exceptional roots. Thus the average
contact count over the Weyl images of this $D_5$ child is
\begin{equation}\label{d43:eq:count}
1,092,000+12(116,235)+6(9,180)+24(495)+12=2,553,792.
\end{equation}
This proves an average count, not yet the count of a fixed child.
The $a_j$ need not be integers because they are orbit averages.

\subsection{Exact count for the specified child}
For the fixed child use the target-preserving subgroup
\[
\Gamma=\langle W(D_5),W(A_3),-1\rangle.
\]
This group preserves the full domain and the condition
$y_1=\cdots=y_5=0$, and has 43 orbits $O_j$.
For averaged multiplicities $v_j$, use rows
\[
A_{\alpha,j}=\sum_{y\in O_j}y^\alpha,
\]
with right-hand side~\eqref{eq:section-moments}, and fix exceptional
orbit values to one. The target is
\[
c_j=|\{y\in O_j:y_1=\cdots=y_5=0\}|.
\]
The supplied certificate selects 43 moment or fixed-root rows, and provides
a rational column vector $w$ such that
\begin{equation}\label{d43:eq:dual}
A^tw=c,\qquad w^tb=2,553,792.
\end{equation}
Thus $|X\cap U^\perp|=c^tv=w^tb$. All inputs depend only on the
root system and spherical moments. Precomposing $\iota$ with a Weyl
element proves the conjugate cases without extending it to $L$.
The verifier regenerates $A,b,c$ from the coordinate model, taking only
the row selection and $w$ from the certificate in
Table~\ref{tab:certificates}.

\subsection{A nested family of fixed coordinate sections}
Extend $\iota$ linearly and define the eight fixed coordinate cuts
\[
 U_k=\iota\bigl(\sqrt3\operatorname{span}\{e_1,\ldots,e_k\}\bigr),
 \qquad 1\le k\le8.
\]
For $k\ge2$ these are the coordinate $D_k$ spans ($D_2=A_1^2$,
$D_3=A_3$); $U_1$ is an axis, not a $D_1$ root subsystem, and
$U_8$ is the whole parent.

\begin{theorem}[Coordinate-section family]\label{thm:e8-coordinate-sections}
Let $L,X,\iota$ satisfy Theorem~\ref{d43:thm:structure}.
For $m=8-k$, the number of minimal vectors perpendicular to the fixed
space $U_k$ is
\begin{equation}\label{eq:e8-coordinate-count}
\begin{split}
 |X\cap U_k^\perp|={}&1092000+18360m+464944\binom m2
                  +11880\binom m3\\
                 &+146880\binom m4+2520\binom m5+31680\binom m6.
\end{split}
\end{equation}
The same holds for every Weyl conjugate of this nested family.
In particular, a Weyl conjugate contains the specified $D_5$ child
of Theorem~\ref{d43:thm:structure} as its rank-five member.
\end{theorem}

\begin{proof}
Use the doubled coordinates $u=2\lambda$ of the complete
nonexceptional domain $D$ above. Let $\Gamma_k$ consist of coordinate
permutations within the first $k$ and last $8-k$ positions and sign
changes in an even number of positions. This subgroup of $W(D_8)$
preserves $D$ and the target $u_1=\cdots=u_k=0$. Its orbits are
specified by the two multisets of squared coordinates and $\prod_i u_i$.
The product distinguishes sign parity when no coordinate is zero;
a zero coordinate removes that restriction.

For a partition $a$, write $m_a$ for the monomial symmetric polynomial
whose distinct monomials have coefficient one. Use all rows
\[
 F_{a,b}(u)=m_a(u_1^2,\ldots,u_k^2)
           m_b(u_{k+1}^2,\ldots,u_8^2),\qquad
 |a|+|b|\le5,
\]
where the partition lengths are at most $k$ and $8-k$, respectively.
For orbit masses $w_j$ set $A_{(a,b),j}=F_{a,b}(u_j)$, with
$u_j\in O_j$. Subtracting the 240 singleton exceptions gives $Aw=b$
for the actual masses, without assuming fibre symmetry. The functionals
$u_i$ have Gram $12I_8$, so for $|a|+|b|=d$ the unadjusted moment is
\begin{equation}\label{eq:e8-coordinate-moment}
 \frac{52416000\,72^d}{48\cdot50\cdots(48+2d-2)}
 \nu_k(a)\nu_{8-k}(b)
 \prod_i(2a_i-1)!!\prod_j(2b_j-1)!!,
\end{equation}
where $\nu_l(a)$ is the number of distinct permutations of $a$
padded with zeros to length $l$.

All eight matrices have full column rank: regenerated square minors
have the nonzero residues in Table~\ref{tab:e8-coordinate-family}.
Thus $w$ is unique over $\Q$. The candidate
\[
 w_j=|O_j|a_r,\qquad
 (a_0,a_2,a_4,a_6,a_8)=(1092000,116235,9180,495,63/4),
 \quad r=\|u_j\|^2/4,
\]
comes from the full Weyl average and satisfies every row exactly.
It must therefore equal the actual orbit masses. This does not assign
the fractional value $63/4$ to individual fibres.

It remains to count the target signatures. Since $k\ge1$, a vector
$\lambda$ perpendicular to $U_k$ has integer coordinates and even
coordinate sum in the remaining $m\le7$ positions. At norms
$0,2,4,6,8$, the numbers of admissible signatures are respectively
\[
 1,\quad 4\binom m2,\quad 2m+16\binom m4,\quad
 24\binom m3+64\binom m6,\quad 160\binom m5.
\]
These come from the signed shapes
$0$, $(1^2)$, $(2),(1^4)$, $(2,1^2),(1^6)$ and $(2,1^4)$.
The root bound excludes $(2,2)$, and $(1^8)$ cannot fit in $m\le7$
coordinates. There are also $4\binom m2$ exceptional roots.
Multiplication by $a_r$ and addition of the exceptions proves
\eqref{eq:e8-coordinate-count}. Full column rank, not radial moments
alone, establishes these counts for the fixed cuts.

For the specified $D_5$ correspondence put $H_4=J_4-2I_4$ and
let $P_{5,8}$ exchange coordinates five and eight. The map
\[
 Q=\tfrac12P_{5,8}\operatorname{diag}(H_4,H_4)
\]
permutes the $E_8$ roots; the certificate expresses it as 31 simple
reflections.
It sends $\alpha_1,\ldots,\alpha_5$ to
\[
 -e_1-e_5,\quad e_3+e_4,\quad e_1-e_2,\quad
 e_2-e_3,\quad e_3-e_4,
\]
which span the first five coordinate directions. Thus $Q^{-1}$ sends
the rank-five cut to the specified $D_5$. No ambient Weyl action is
required.
\end{proof}

\begin{table}[ht]
\centering
\caption{Exact fixed coordinate-section counts and rank certificates.
The minors use every orbit column and the specified subset of moment rows.}
\label{tab:e8-coordinate-family}
\begin{tabular}{rrrrr}
\toprule
$k$&Moment rows&Orbit columns&Minor modulo $1000003$&$|X\cap U_k^\perp|$\\\midrule
1&45&25&910777&16815624\\
2&62&34&790744&10663920\\
3&70&39&894205&6688960\\
4&72&41&314211&4149504\\
5&70&39&401115&2553792\\
6&62&34&244188&1593664\\
7&45&25&296071&1110360\\
8&19&13&272057&1092000\\\bottomrule
\end{tabular}
\end{table}

\subsection{What accounts for the improvement}
The comparison in~\cite{crosssections} has $2,545,056$ contacts from
a primitive $\sqrt3D_5$ in $P_{48p}$. Both sections have determinant
972, singleton cluster, $\beta=0$ and permutation-isometric Grams.
The gain therefore lies in the embedding: the specified $E_8$ parent
in $P_{48n}$ forces the larger perpendicular shell, despite identical
intrinsic section data.

\section{A moment-certified projection in dimension 45}\label{sec:projection45}
Projecting three fibres of a quadratic-residue neighbour gives 7380090
points. A design-moment identity expresses their total size through one
fibre, which an anchor graph counts exactly.

\subsection{The specified quadratic-residue neighbour}\label{sec:qr-neighbour}
Let $C$ be the extended ternary quadratic-residue code of length 48.
We use its published parameters $[48,24,15]$~\cite{leech-sloane1971,nebe}.
The supplied generator fixes a degree-23 factor of
$1+x+\cdots+x^{46}$ over $\F_3$, its coordinate order and a self-dual
extension coordinate.
Set
\[
M=\{z\in\Z^{48}:z\bmod3\in C\},\qquad L=M/\sqrt3,
\qquad L_0=\{z/\sqrt3\in L:z^2\equiv0\pmod6\}.
\]
An integer basis $B_0$ for $M$ and a characteristic numerator
$\chi\in M$ are supplied. They satisfy
\[
\chi_i\text{ odd for all }i,\qquad
\chi^2=264,\qquad
\chi\cdot z\equiv z^2\pmod6\quad(z\in M).
\]
The last congruence is checked on $B_0$ and extends by integrality.
Define
\[
\Lambda=L_0\cup\bigl(\chi/(2\sqrt3)+L_0\bigr).
\]
The characteristic relation and $\chi^2/12=22$ make $\Lambda$ even
and integral. The two index-two constructions give covolume one.
Integer row operations on $B_0$ and the two cosets reconstruct a basis
$B$ with $BB^t=12G$, where $G$ is the Gram of the 298-vector certificate;
the verification compares all entries.

\begin{lemma}\label{lem:qr-minimum}
The specified lattice $\Lambda$ has minimum squared norm six.
\end{lemma}
\begin{proof}
In $L_0$, a putative vector of squared norm two or four has numerator
$z$ with $z^2\le12$. The minimum code weight forces $z\bmod3=0$.
Its nonzero squared norm is then $3w^2$, whose positive even values
are at least six.

In the other coset, the numerator $y=\chi+2z$ has all coordinates
odd, so $y^2/12\ge4$. Equality requires $y_i=\pm1$.
Let $s_i\in\{-1,1\}$ represent $\chi_i\bmod3$. The certificate
checks that $s\in C$ and $\chi\cdot s=66$. Multiplying coordinates
by $s$ yields an equivalent extremal ternary code containing the
all-one word. Its binary-valued words have weights $0,24,48$, by
the two-symbol specialization of the complete weight
enumerator~\cite[equation (20)]{munemasa-tamura}.
Write $y=s(1-2b)$ coordinatewise. Since $y\bmod3\in C$, the
binary word $b$ belongs to the normalized code, and its weight $w$
is one of these three values. But $\chi_i s_i\equiv1\pmod6$,
so
\[
\chi\cdot z
=\frac{\chi\cdot s-\chi^2}{2}-\sum_i\chi_i s_i b_i
\equiv-99-w\equiv3\pmod6.
\]
This contradicts $z/\sqrt3\in L_0$. Evenness excludes other
positive squared norms below six, and $\sqrt3(e_0+e_1)\in L_0$ has
squared norm six.
\end{proof}

Thus the specified Gram has a minimal shell of size $52416000$
by~\eqref{eq:theta48}, and this shell is an 11-design~\cite{nebe-designs}.

\subsection{A complete signature domain and its moments}
For the counting argument let $\Lambda$ be any extremal even unimodular
lattice of rank 48, and let $X$ be its squared-norm-six shell. Thus
$|X|=52416000$ by~\eqref{eq:theta48}, and $X$ is an 11-design.
Suppose $t_1,t_2,t_3\in\Lambda$ have Gram matrix
\begin{equation}\label{eq:45gram}
H=8I_3-2J_3
=\begin{pmatrix}6&-2&-2\\-2&6&-2\\-2&-2&6\end{pmatrix},
\quad H^{-1}=(I_3+J_3)/8,\quad\det H=128.
\end{equation}
For $x\in X$, put $y_i=\ip{x}{t_i}$ and write $f(y)$ for the
number of vectors with that signature. No primitivity hypothesis is
needed. The vectors $t_1,t_2,t_3,-t_1-t_2-t_3$ form a regular
tetrahedron of squared radius six. They and their negatives give the
eight exceptional signatures
\[
E=\{\pm He_1,\pm He_2,\pm He_3,\pm H(1,1,1)^t\},
\]
each of multiplicity one. Applying the minimum to these eight vectors,
as in Proposition~\ref{prop:section-statistics}, puts every other
signature in
\[
D_0=\{y\in\Z^3:|y_i|\le3,\ |y_1+y_2+y_3|\le3\}.
\]
Writing $y_4=-y_1-y_2-y_3$ gives $q(y)=\sum_{i=1}^4y_i^2/8\le9/2$.
There are $\binom{15}{3}-4\binom{8}{3}=231$ such quadruples, so
$D_0\cup E$ is the original 239-signature domain.

Each $t_i+t_j$ has squared norm eight, so the minimum applied to
$x\pm(t_i+t_j)\ne0$ also gives
\begin{equation}\label{eq:45longer-anchors}
|y_i+y_j|\le4\qquad(1\le i<j\le3).
\end{equation}
These inequalities exclude exactly the $6+24=30$ permutations of
$(3,3,-3,-3)$ and $(3,2,-3,-2)$ from $D_0$. The refined domain $D$
has 201 signatures and 101 antipodal orbits; all eight exceptions remain.

Apply~\eqref{eq:section-moments} to all even total degrees at most ten
and subtract $E$. On $D_0$ this gives $Av=b$ with 161 rows and 116
variables $v_{[y]}=f(y)=f(-y)$. Its entries are
$A_{\alpha,[y]}=(2-\mathbf1_{y=0})y^\alpha$, including monomials
with odd individual exponents. The system depends only on $H$.

Let $c$ be the coefficient vector of
\[
f(e_1)+f(e_2)+f(e_3)-9f(-2,-2,3).
\]
The stored rational multiplier $\lambda$ satisfies
\begin{equation}\label{eq:45dual}
c-A^t\lambda\ge0,\qquad \lambda^tb=7377408.
\end{equation}
All seven nonzero slack entries belong to the 15 antipodal orbits
excluded by~\eqref{eq:45longer-anchors}. Restricting to $D$ therefore gives
\[
A_D^t\lambda=c_D,\qquad\lambda^tb=7377408.
\]
The excluded fibres vanish, so the original lower certificate sharpens
to the exact identity
\begin{equation}\label{eq:45count}
f(e_1)+f(e_2)+f(e_3)=7377408+9f(-2,-2,3).
\end{equation}
The coefficient vector uses $1$ on $[e_i]$ and $-9$ on $[(-2,-2,3)]$,
equivalently half these weights at each antipode in
Proposition~\ref{prop:section-statistics}. The checker regenerates both
domains and all rational identities; individual fibre sizes need not be
determined.

\subsection{Witness amplification}
\begin{theorem}\label{thm:45-transfer}
Let $\Lambda$ be an extremal even unimodular lattice of rank 48,
and let $t_1,t_2,t_3\in\Lambda$ have Gram matrix~\eqref{eq:45gram}.
With $f$ as above, projection onto
$\operatorname{span}_{\R}\{t_1,t_2,t_3\}^\perp$ and normalization
send the three fibres $e_1,e_2,e_3$ injectively to a kissing
configuration with exactly $7377408+9f(-2,-2,3)$ points.
In particular, $m$ distinct vectors of squared norm six and signature
$(-2,-2,3)$ prove
\[
K(45)\ge7377408+9m.
\]
\end{theorem}
\begin{proof}
Each projected squared norm is $6-1/4=23/4$.
Two original distinct vectors have inner product at most three.
Their projected products are therefore at most $3-1/4=11/4$ within
one fibre and $3-1/8=23/8$ in different fibres. This also excludes
collisions. Normalization gives a kissing configuration in dimension
45, whose cardinality is~\eqref{eq:45count}. Finally
$f(-2,-2,3)\ge m$ for any such witness set.
\end{proof}

\subsection{A regular anchor graph}
For $p\in\Lambda$ of squared norm eight, define
$A_p=\{x\in X:\ip{x}{p}=4\}$. The involution $x\mapsto p-x$
pairs its elements. Let $\mathcal G_p$ have these pairs as vertices,
with two distinct vertices adjacent when representatives have inner
product one or three.

\begin{proposition}\label{prop:45-anchor-graph}
For every extremal even unimodular $\Lambda$ of rank 48 and every
squared-norm-eight $p\in\Lambda$, the graph $\mathcal G_p$ is well
defined, has 1128 vertices and is regular of degree 368.
For nonadjacent distinct vertices represented by $r,s$, the section
\[
(t_1,t_2,t_3)=(-s,s-p,p-r)
\]
has Gram matrix~\eqref{eq:45gram}, and its entire fibre satisfies
\begin{equation}\label{eq:45-fibre-graph}
f(-2,-2,3)=|N(\bar r)\setminus N(\bar s)|
         =368-|N(\bar r)\cap N(\bar s)|.
\end{equation}
Conversely every ordered section with this Gram matrix has such an
anchor representation.
\end{proposition}
\begin{proof}
Integrality and the minimum bound give $|\ip{x}{p}|\le4$ for
$x\in X$. For $P(t)=t^2(t^2-1)(t^2-4)(t^2-9)$, the even moments
through degree eight give
\[
2P(4)|A_p|=90961920,\qquad P(4)=20160,\qquad |A_p|=2256.
\]
The partner $p-x\ne x$ again belongs to $A_p$. For endpoints $r,a$
from distinct pairs, the minimum gives
$\ip{r}{a}\le3$ and $4-\ip{r}{a}\le3$, hence products in
$\{1,2,3\}$. Exchanging an endpoint replaces the product by four
minus itself, proving well-definedness and the vertex count.

To obtain the degree and the frame identity below, take any
$z\in p^\perp$. The mixed design moments are
\[
M_k(z):=\sum_{x\in X}\ip{p}{x}^{2k}\ip{z}{x}^2
=\frac{52416000\,6^{k+1}8^k(2k-1)!!}
       {48\cdot50\cdots(48+2k)}\,\|z\|^2\qquad(1\le k\le4).
\]
Consequently
\begin{equation}\label{eq:45-filtered-second}
\sum_{x\in A_p}\ip{z}{x}^2
=\frac{M_4(z)-14M_3(z)+49M_2(z)-36M_1(z)}{2P(4)}
=192\|z\|^2.
\end{equation}
The degree is at most ten. For $u=r-p/2$, $u^2=4$ and
$\ip{u}{x}=\ip{r}{x}-2$. In~\eqref{eq:45-filtered-second},
$r,p-r$ contribute $32$, each adjacent pair contributes two and
the others zero. Thus $768=32+2\deg(\bar r)$.

For a nonadjacent pair $\ip{r}{s}=2$, which gives the asserted Gram
matrix directly. Each vertex in $N(\bar r)\setminus N(\bar s)$
has a unique endpoint $a$ with $\ip{r}{a}=3$; it satisfies
$\ip{s}{a}=2$. The vector $w=p-a$ has squared norm six and
signature $(-2,-2,3)$. Conversely that signature forces
$\ip{s}{w}=2$, $\ip{p}{w}=4$ and $\ip{r}{w}=1$.
Thus $a=p-w$ recovers the unique graph vertex, proving the bijection.

Conversely, from any ordered section with Gram~\eqref{eq:45gram}, take
$p=-t_1-t_2$, $s=-t_1$, and $r=-t_1-t_2-t_3$.
Their products are $p^2=8$, $r^2=s^2=6$,
$\ip{p}{r}=\ip{p}{s}=4$ and $\ip{r}{s}=2$, giving the
required representation.
\end{proof}

For a fixed anchor and the specified Gram, minimizing the codegree
therefore maximizes the exact projected count
\[
7377408+9\bigl(368-|N(\bar r)\cap N(\bar s)|\bigr)
=7380720-9|N(\bar r)\cap N(\bar s)|.
\]

\subsection{Signed spectrum and lattice restrictions}
Choose one endpoint $r_i$ from each of the 1128 pairs and set
$u_i=r_i-p/2$. Let $S_{ii}=0$ and $S_{ij}=\ip{u_i}{u_j}$ for
$i\ne j$, so $S_{ij}\in\{0,\pm1\}$ and the unsigned adjacency
matrix is $A=(|S_{ij}|)$.

\begin{proposition}\label{prop:45-signed-anchor}
Under the hypotheses of Proposition~\ref{prop:45-anchor-graph},
\begin{equation}\label{eq:45-signed-spectrum}
\sum_{i=1}^{1128}u_i\otimes u_i=96I_{p^\perp},\qquad
S^2=88S+368I,\qquad
\operatorname{Spec}(S)=\{92^{(47)},(-4)^{(1081)}\}.
\end{equation}
For distinct $i,j$, write $c_{ij}=|N(i)\cap N(j)|$ and let
$c_{ij}^{\pm}$ count common neighbours $k$ with
$S_{ik}S_{kj}=\pm1$. Then
\begin{equation}\label{eq:45-signed-codegrees}
c_{ij}^{\pm}=\frac{c_{ij}\pm88S_{ij}}2.
\end{equation}
In particular, all codegrees are even, nonadjacent pairs have equally
many common neighbours of either sign, and adjacent pairs have
$c_{ij}\ge88$.

The unsigned matrix satisfies $A+16I\succeq0$. Moreover, no four
of the $u_i$ can be signed so that all six mutual products are $-1$.
Changing the chosen endpoints conjugates $S$ by a diagonal sign
matrix; it preserves its spectrum, the unsigned graph and this
four-vertex obstruction.
\end{proposition}
\begin{proof}
Dividing~\eqref{eq:45-filtered-second} by two gives
$\sum_i\ip{z}{u_i}^2=96\|z\|^2$ for $z\in p^\perp$,
and polarization proves the frame identity. Hence the Gram matrix
$K=4I+S$ has rank 47 and satisfies $K^2=96K$, proving the
quadratic identity and spectrum.
For $i\ne j$, the entry $(S^2)_{ij}$ is
$c_{ij}^+-c_{ij}^-=88S_{ij}$, which proves
\eqref{eq:45-signed-codegrees}.

The Gram matrix of the rank-one tensors $u_i\otimes u_i$, with
their usual inner product, has entries $\ip{u_i}{u_j}^2$ and is
$16I+A$. This proves the unsigned spectral bound, not a two-valued
spectrum for $A$. If four signed vectors $\varepsilon_i u_i$ had
all mutual products $-1$, their sum would be a lattice vector:
\[
v=\sum_{i=1}^4\varepsilon_i r_i
  -\frac{\sum_{i=1}^4\varepsilon_i}{2}\,p\in\Lambda,
\qquad v^2=4\cdot4-2\binom42=4.
\]
This contradicts the lattice minimum. Finally, replacing $r_i$ by
$p-r_i$ negates $u_i$, giving the asserted diagonal conjugation.
\end{proof}

For every section represented by a nonadjacent pair,
\begin{equation}\label{eq:45-count-congruence}
N=7380720-9c_{ij}\equiv0\pmod{18}.
\end{equation}
Improvement on the count at $c_{ij}=70$ requires $c_{ij}\le68$
and hence $N\ge7380108$, together with the lattice restrictions above.

\subsection{The certified anchor and its optimum}
The earlier Pless-lattice route found 292 witnesses. The 298-vector
construction instead uses an anchor in the odd-coordinate sector of
the QR neighbour: in numerator coordinates it has shape
$(3^6,1^{42})$ up to signs and permutation. The parent certificate
specifies the anchor, chosen endpoints and 2256 distinct valid
elements of $A_p$; the moment count proves that this list is complete.
Exact checking of all 428076 unordered nonadjacent pairs gives codegrees
70 to 148, with a unique pair attaining 70. It recovers the stored
section $T$ and complete 298-vector fibre. Thus 7380090 is optimal for
this fixed anchor and section Gram, not over other anchors, lattices or
Grams. Integer checks also verify the frame and signed-graph identities
and, independently,
\[
TGT^t=H,\qquad wGw^t=6,\qquad TGw^t=(-2,-2,3)^t
\]
for the 298 distinct rows $w$. The supplied $WGT^t=I_3$ certifies
primitivity of this section, although the results do not require it.

\begin{corollary}\label{thm:d45}
The specified QR-neighbour construction gives a kissing configuration
in dimension 45 with exactly 7380090 points. In particular,
$K(45)\ge7380090$.
\end{corollary}
\begin{proof}
Lemma~\ref{lem:qr-minimum} establishes the ambient hypotheses.
The complete parent graph and Proposition~\ref{prop:45-anchor-graph}
give $f(-2,-2,3)=368-70=298$.
Theorem~\ref{thm:45-transfer} therefore gives
$7377408+9(298)=7380090$ points.
\end{proof}

\section{Finite search and mathematical structure}\label{sec:research}
The constructions expose different freedoms in a constrained configuration:
moving a contact layer, changing signed supports, using another lattice
shell, or selecting less-overlapping images. Their limits are also
explicit. Two cap changes are necessary for the fixed insertion in
dimension 25; the next Leech shell cannot augment the unchanged
dimension-38 base; and pure tails force the heights and disjoint head
assignments in the paired lifts.

The section counts use a different reduction. Moments constrain all
admissible projection signatures at once. The embedded parent forces
the dimension-43 count, whereas the dimension-45 section vectors remove
every signature with positive dual slack, yielding the exact
identity~\eqref{eq:45count}. These arguments identify the required
statistic without determining every fibre. Together with the finite
witnesses, they yield the nineteen stated lower bounds.

\section{Finite verification}\label{sec:certificates}
Each lower bound has two parts: a mathematical reduction covering
all types of point pairs, and a finite witness satisfying the reduction's
hypotheses. The supplement contains these witnesses and programs
for checking the required identities and inequalities.
Construction data, verification programs and mathematical research accounts
are available in the project repository:
\begin{center}\small
\url{https://github.com/Oxelra-AI/Qiushi-Engine-Kissing-Number-Research}.
\end{center}
The \href{https://github.com/Oxelra-AI/Qiushi-Engine-Kissing-Number-Research/tree/2398ab45306899f6ceb1b361a8b67097501ebbeb}{fixed supplementary-materials snapshot}
contains the data and programs used in this article.
Table~\ref{tab:certificates} identifies the finite objects used in the
proofs; paths are relative to \qpath{constructions/}.

\begin{table}[ht]
\centering\small
\caption{Finite inputs to the construction proofs.}
\label{tab:certificates}
\begin{tabular}{@{}p{.13\textwidth}p{.57\textwidth}p{.23\textwidth}@{}}\toprule
Dimensions & Mathematical input & Location\\\midrule
25 & Head data, the 552-point layer and two repair vectors & \qpath{d25/}\\
27 & Five 496-point blocks; first-layer heads and 311 labelled second-layer owners
   & \qpath{leech/}, \qpath{d27/}\\
32--34, 37, 39 & Complete weight-eight support lists, parent families and exchanges
   & \qpath{codes/data/}\\
35--37 & Signed local codes and their boundary embeddings & \qpath{codes/}\\
38 & Shell partition, tail rotation and 144 added lines & \qpath{d38/}\\
49--55 & The 7077-line class, isometries and tail assignments & \qpath{p48/}\\
43 & Parent embedding, orbit moments and fixed-cut rank certificates & \qpath{sections/d43/}\\
45 & QR-neighbour basis, moment multipliers, parent graph and 298 witnesses
   & \qpath{sections/qr/}, \qpath{sections/d45/}\\\bottomrule
\end{tabular}
\end{table}

For dimension 25, the checker first certifies the 197579-point input
from its integer and rational head data. The 552-point motion has
symbolically preserved equalities, as proved in
Section~\ref{sec:motion25}. The remaining strict inequalities use
192-bit outward-rounded Arb intervals or integer comparisons after
squaring positive quantities. The two repair vectors are given as
integer rows with exact normalization.

For the triangular Leech lifting, the inputs are five 496-point blocks, the Golay
coordinate convention, and exact tail coordinates. Integer products,
membership, disjointness and radical identities establish the hypotheses
of the lifting theorem. The separate implementations are
\qpath{leech/verify.py} and \qpath{leech/verify_independent.py}.
The updated dimension-27 construction additionally
checks all first-layer owners, the 311 second-layer owners and labels,
and every interlayer inequality in Section~\ref{sec:leech-two-layer}.
For the layered constructions, support intersection
checks are combined with the distances of the dense binary code or its
coset unions. The five supplied parent families also reproduce the
support exchanges, including the exact union of deleted conflicts.
The signed replacements additionally check every new signed pair and
every boundary support intersection; in dimension 37 the checker
regenerates all 512 vectors from the 32 transverse four-subsets.

For the bounds in dimensions 49--55, the checker verifies the
7077-line mother class in the catalogue Gram basis, the exact image
matrices, the seven union counts, and an injective assignment to
antipodal tail lines. The published minimum and the theta-series
identity supply the mother-shell premise. The same images applied
to the original 7069-line class verify every term in
Table~\ref{tab:p48-gain}. The ambient minimum supplies the
cross-block inequality.

For dimension 38, the checker regenerates the Leech shell from
the base package's Golay generator, checks the class partition,
tail triples and rational rotation, and verifies the 144 new
norm-48 lines against both the complete shell and one another.
All pair products in the additional tests are integers. Two
implementations check the 144-line extension.

The dimension-43 verifier checks the concrete lattice embedding, then
regenerates the 26641-signature domain, subgroup action, moment matrix
and target vector from standard $E_8$ coordinates. It checks the supplied
rational multiplier against this regenerated system. These latter steps
use only the hypotheses of Theorem~\ref{d43:thm:structure}.
The coordinate-family checker additionally constructs all eight
fixed-cut moment systems and their full-rank minors. Independent integer
determinants and the explicit Weyl map check the rank certificates and
the correspondence with the original $D_5$.
The dimension-45 checks reconstruct the
specified QR neighbour, verify its exact basis Gram, establish the
norm-four parity obstruction, and check all 298 fibre witnesses.
The parent-graph calculation checks all 2256 endpoints, their pairing
about $p/2$, the universal degree from the tenth moments, and equality
between the complete target fibre and the 298 witnesses.
It also verifies the signed-frame identity and enumerates all 428076
nonadjacent pairs to determine the optimum for this fixed anchor.
A separate rational calculation regenerates the 161-row moment system
and checks the supplied multiplier in~\eqref{eq:45dual} directly.
The norm-eight restrictions reduce the 239-signature domain to 209
signatures; on its 101 nonexceptional antipodal orbits the multiplier
has zero slack. The independent standard-library checker
\qpath{sections/d45/verify_statistics.py} reconstructs this equality
and the anchor moments over the rationals.

\subsection{Autonomous mathematical research}\label{sec:autonomous-research}
The nineteen constructions arose from Qiushi Engine's autonomous
mathematical research, encompassing construction searches, mathematical
analysis and computational verification. Its research connected
changes of representation with new finite objects: coordinated layer
motion (Section~\ref{sec:motion25}), joint label selection
(Section~\ref{sec:leech-two-layer}), signed-design transfer
(Section~\ref{sec:signed}), and moment functionals adapted to the
section counts (Sections~\ref{sec:sections} and~\ref{sec:projection45}).
The accompanying research
accounts describe these mathematical developments and identify their
inputs and resulting constructions. Background on the research system
is given in~\cite{qiushi-optics}.

\appendix
\section{The integer Leech shell and tail conventions}\label{app:leech}

This appendix fixes the normalization of the classical Golay
description of the Leech shell~\cite{conway1999sphere}. Let
$G\subset\F_2^{24}$ be generated by the twelve words in
Table~\ref{tab:golay}, with coordinate zero at the left. Enumeration
and binary elimination give
\begin{equation}\label{eq:golay}
 \dim G=12,\qquad G\subseteq G^\perp,\qquad
 W_G(z)=1+759z^8+2576z^{12}+759z^{16}+z^{24}.
\end{equation}
An \emph{octad} is the support of a weight-eight word of $G$.
Two distinct octads intersect in at most four places, because their
symmetric difference is a nonzero word of weight at least eight.
For $c\in G$ put $\eps_i(c)=(-1)^{c_i}$ and define
\begin{align}
 \Sigma_A&=\{4\eps e_i+4\eps'e_j:0\le i<j<24,\ \eps,\eps'\in\{\pm1\}\},
 \label{eq:shell-A}\\
 \Sigma_B&=\left\{2\sum_{i\in O}\eta_i e_i:
       O\text{ an octad},\ \eta_i\in\{\pm1\},\ \prod_{i\in O}\eta_i=1\right\},
 \label{eq:shell-B}\\
 \Sigma_C&=\left\{\sum_{i=0}^{23}\eps_i(c)e_i-4\eps_j(c)e_j:
                         c\in G,\ 0\le j<24\right\}.
 \label{eq:shell-C}
\end{align}
Set $\Sigma=\Sigma_A\cup\Sigma_B\cup\Sigma_C$.

\begin{lemma}\label{lem:shell}
This set satisfies~\eqref{eq:leech-shell}.
\end{lemma}
\begin{proof}
The absolute coordinate patterns distinguish the three families.
Support and signs recover the parameters of the first two. In the
third, the unique coordinate of absolute value three identifies $j$,
after which the signs recover $c$. Thus the cardinalities are
\[
 |\Sigma_A|=4\binom{24}{2}=1104,\quad
 |\Sigma_B|=759\cdot128=97152,\quad
 |\Sigma_C|=4096\cdot24=98304.
\]
They sum to $196560$. The respective squared norms are
$16+16$, $8\cdot4$, and $9+23$, all equal to $32$.

For two $A$-vectors, different supports give at most one contribution
of $16$, and equal supports with different signs give at most zero.
For $A$--$B$ pairs there are at most two contributions of eight;
for $A$--$C$ pairs the bound is $4(3+1)=16$.
For $B$--$B$ pairs on different octads, the intersection bound gives
at most $4\cdot4=16$. On the same octad, different even sign patterns
differ at least twice, giving at most $4(8-4)=16$.

For a $B$--$C$ pair, let $O$ be the octad and $j$ the exceptional
coordinate of the $C$-vector. If $j\notin O$, the inner product is
at most $2|O|=16$. Otherwise put $y_i=\eta_i\eps_i(c)$ on $O$.
Self-orthogonality of $G$ and the even sign condition imply
$\prod_{i\in O}y_i=1$, and the inner product is
\[
 2\sum_{i\in O}y_i-8y_j.
\]
If $y_j=1$, this is at most eight. If $y_j=-1$, at least one other
$y_i$ is negative, so the sum is at most four and the expression is
at most sixteen.

Finally consider $C$--$C$ pairs with words $c,c'$ and exceptional
coordinates $j,k$. Set $s_i=(-1)^{c_i+c'_i}$ and
$S=\sum_i s_i=24-2\wt(c+c')$. If $j=k$, the inner product is
$S+8s_j$; distinct vectors then require $c\ne c'$, giving $S\le8$.
If $j\ne k$, it is $S-4(s_j+s_k)$. This equals sixteen when
$c=c'$, and is at most $8+8$ otherwise. These six family pairs
exhaust all cases.
\end{proof}

\begin{table}[ht]
\centering\small
\caption{A binary generator for the Golay code in~\eqref{eq:golay}.
Coordinates run from zero to $23$ from left to right.}
\label{tab:golay}
\begin{tabular}{rl}
\toprule
Row & Binary word\\
\midrule
0 & \texttt{010111100101000000000001}\\
1 & \texttt{110111001010000000000010}\\
2 & \texttt{100011011101000000000100}\\
3 & \texttt{000110111110000000001000}\\
4 & \texttt{011101011100000000010000}\\
5 & \texttt{111001101001000000100000}\\
6 & \texttt{001001110111000001000000}\\
7 & \texttt{011010001111000010000000}\\
8 & \texttt{101100011011000100000000}\\
9 & \texttt{101111110000001000000000}\\
10 & \texttt{110010110011010000000000}\\
11 & \texttt{111100100110100000000000}\\
\bottomrule
\end{tabular}

\end{table}

For compatibility with the supplied tail arrays, label $V$ by
\begin{equation}\label{eq:tail-labels}
\begin{array}{c|rrr@{\qquad}c|rrr}
i&\multicolumn{3}{c}{v_i}&i&\multicolumn{3}{c}{v_i}\\\hline
0&-1&0&-1&6&-1&0&1\\
1&0&-1&-1&7&1&-1&0\\
2&-1&-1&0&8&0&-1&1\\
3&0&1&-1&9&1&1&0\\
4&-1&1&0&10&0&1&1\\
5&1&0&-1&11&1&0&1.
\end{array}
\end{equation}
The triangles in~\eqref{eq:triangles} then have labels
$\{0,7,10\}$, $\{1,6,9\}$, $\{2,3,11\}$, and $\{4,5,8\}$.
For the pure tails set
\[
 (\sigma(0),\ldots,\sigma(11))=(0,3,4,1,5,2,9,6,10,7,8,11),
 \qquad p_a=Rv_{\sigma(a)}/\sqrt2.
\]
The source rows $\widehat t_i,\widehat p_a$ in
\qpath{constructions/leech/data/tails.json}
are related to these coordinates by right multiplication by
\begin{equation}\label{eq:tail-map}
 O=\begin{pmatrix}
 \sqrt2/2&0&\sqrt2/2\\
 -\sqrt6/6&\sqrt6/3&\sqrt6/6\\
 -\sqrt3/3&-\sqrt3/3&\sqrt3/3
 \end{pmatrix},\qquad O^{\mathsf T}O=I.
\end{equation}
In row-vector notation,
$\widehat t_iO=v_i^{\mathsf T}/\sqrt6$ and
$\widehat p_aO=v_{\sigma(a)}^{\mathsf T}R^{\mathsf T}/\sqrt2$.
The exact coordinate-change test checks these identities for every
label, so the common rotation preserves both within-set and cross-set
inner products.

\section{Binary generators and coset representatives}\label{app:codes}

Table~\ref{tab:generators} gives the generators used for the three
dense sign codes. An entry is a hexadecimal mask, representing the
binary word $c$ by $\sum_i c_i2^i$. The generator of length $32$ is
used unchanged in ambient dimensions $32$--$37$. The
length-$38$ and length-$39$ columns are the kernels reconstructed
from~\eqref{eq:qr}. Table~\ref{tab:representatives} supplies the affine
representatives; the length-$32$ code uses only zero.

\begin{table}[ht]
\centering\small
\caption{Generator masks for $H_{32}$, $H_{38}$, and $H_{39}$. Blank
entries reflect their different dimensions.}
\label{tab:generators}
\begin{tabular}{rrrr}
\toprule
Row & $H_{32}$ & $H_{38}$ & $H_{39}$\\
\midrule
0 & \texttt{4EF0C081} & \texttt{2000C0DE0B} & \texttt{4000152727}\\
1 & \texttt{24A09082} & \texttt{1000AA229D} & \texttt{2000C0DE0B}\\
2 & \texttt{48609084} & \texttt{8009F5CD6} & \texttt{1000AA229D}\\
3 & \texttt{1E505088} & \texttt{4004FAE6B} & \texttt{8009F5CD6}\\
4 & \texttt{6C905090} & \texttt{200ED9AAD} & \texttt{4004FAE6B}\\
5 & \texttt{FA30C0A0} & \texttt{100BC80CE} & \texttt{200ED9AAD}\\
6 & \texttt{AAC000C0} & \texttt{805E4067} & \texttt{100BC80CE}\\
7 & \texttt{D2608100} & \texttt{40E56DAB} & \texttt{805E4067}\\
8 & \texttt{7E60D200} & \texttt{20B8FB4D} & \texttt{40E56DAB}\\
9 & \texttt{B8601400} & \texttt{1096303E} & \texttt{20B8FB4D}\\
10 & \texttt{D800D800} & \texttt{84B181F} & \texttt{1096303E}\\
11 & \texttt{CC606000} & \texttt{4EFC197} & \texttt{84B181F}\\
12 & \texttt{1EE10000} & \texttt{2BDAD53} & \texttt{4EFC197}\\
13 & \texttt{B2B20000} & \texttt{1949B31} & \texttt{2BDAD53}\\
14 & \texttt{74740000} &  & \texttt{1949B31}\\
15 & \texttt{D8D80000} &  & \\
16 & \texttt{FF000000} &  & \\
\bottomrule
\end{tabular}

\end{table}

\begin{table}[ht]
\centering\small
\caption{Representatives of the eleven cosets in length $38$ and the
ten cosets in length $39$.}
\label{tab:representatives}
\begin{tabular}{rrr}
\toprule
$a$ & $r_a$ in length 38 & $r_a$ in length 39\\
\midrule
0 & \texttt{0} & \texttt{0}\\
1 & \texttt{C6E94} & \texttt{ADCFB}\\
2 & \texttt{DF6ABF} & \texttt{F7795F}\\
3 & \texttt{27B441} & \texttt{740E29}\\
4 & \texttt{D3CE66} & \texttt{940DFF}\\
5 & \texttt{41BBBB} & \texttt{3B66AB}\\
6 & \texttt{BC1E1F} & \texttt{6CD301}\\
7 & \texttt{52727} & \texttt{4709D4}\\
8 & \texttt{9951C8} & \texttt{98A926}\\
9 & \texttt{7001EB} & \texttt{D23688}\\
10 & \texttt{3F6969} & \\
\bottomrule
\end{tabular}

\end{table}

The complete kernel weight distributions are given in
Table~\ref{tab:weights}. The least nonzero weights give the within-coset
distance bounds. For every distinct pair of representatives in
Table~\ref{tab:representatives}, the affine space
$r_a+r_b+H_m$ has minimum weight ten. Enumeration of these spaces
establishes both the cross-coset distances and disjointness of the
cosets. This determines each code in Proposition~\ref{prop:signs}
entirely from the printed data.

\begin{table}[ht]
\centering\small
\caption{Weight distributions $A_w=|\{h\in H_m:\wt(h)=w\}|$.
Every omitted weight has coefficient zero.}
\label{tab:weights}
\begin{tabular}{rrrr}
\toprule
$w$ & $H_{32}$ & $H_{38}$ & $H_{39}$\\
\midrule
0 & 1 & 1 & 1\\
8 & 908 & 0 & 0\\
10 & 3328 & 0 & 0\\
12 & 14784 & 669 & 970\\
14 & 27392 & 0 & 0\\
16 & 38246 & 5289 & 8952\\
18 & 27392 & 0 & 0\\
20 & 14784 & 8002 & 16452\\
22 & 3328 & 0 & 0\\
24 & 908 & 2310 & 5988\\
28 & 0 & 113 & 402\\
32 & 1 & 0 & 3\\
\bottomrule
\end{tabular}

\end{table}

\section{An explicit Pless neighbour and its minimal shell}\label{app:pless-neighbour}

This appendix establishes the minimum and shell size in an explicit
integer model of an even neighbour. The theta-series argument applies
to every extremal even unimodular lattice of rank 48, including the
catalogue lattice in Section~\ref{app:p48} and the QR neighbour in
Section~\ref{sec:projection45}.

\subsection{The specified even neighbour}
Index the rows and columns of a $24\times24$ matrix $S$ by
$\{\infty\}\cup\F_{23}$, in that order. With $\chi$ the quadratic
character modulo $23$, set
\begin{equation}\label{eq:paley}
 S_{\infty\infty}=0,\quad S_{\infty j}=1,\quad S_{i\infty}=-1,
 \quad S_{ij}=\chi(i-j)\quad(i,j\in\F_{23}).
\end{equation}
The identity $SS^{\mathsf T}=23I$ follows from the quadratic-character
sum $\sum_x\chi(i-x)\chi(j-x)=-1$ for $i\ne j$. Consequently the
ternary code
\[
 C_3=\{(u,uS):u\in\F_3^{24}\}\subset\F_3^{48}
\]
is self-dual: its systematic generator has rank $24$, and
$I+SS^{\mathsf T}=24I=0$ over $\F_3$. Put
\begin{equation}\label{eq:construction-A}
 M=\{z\in\Z^{48}:z\bmod3\in C_3\},\qquad L=M/\sqrt3.
\end{equation}
The row basis
\[
 \begin{pmatrix}I_{24}&S\\0&3I_{24}\end{pmatrix}
\]
has determinant $3^{24}$. The code orthogonality makes $L$ integral,
and the scaling gives covolume one. It is odd, since it contains
$3e_0/\sqrt3$ of squared norm three.

Let $L_0=\{x\in L:x\cdot x\in2\Z\}$ and define
\begin{equation}\label{eq:characteristic}
 c=(1^{24},-23,1^{23}),\qquad g=c+6e_0,
 \qquad L'=L_0\cup\left(\frac{g}{2\sqrt3}+L_0\right).
\end{equation}
The characteristic congruence
$c\cdot z/3\equiv z\cdot z/3\pmod2$ for $z\in M$ is checked on
the displayed basis. This suffices since the squared norm modulo two
is additive on an integral lattice. Also $g/\sqrt3\in L_0$,
$g\cdot g/3=208$, and $g/(2\sqrt3)$ has squared norm $52$.
Its product with $z/\sqrt3\in L_0$ is integral by the characteristic
congruence. Thus~\eqref{eq:characteristic} is an even lattice, obtained
by extending the index-two sublattice $L_0$ by two cosets, and has
covolume one.

The matrix~\eqref{eq:paley} is the defining matrix of the Pless symmetry
code $C(23)$, with the same signs and coordinate order as
\cite[Section~1]{tonchev}. In particular, $C_3$ has minimum Hamming
weight $15$~\cite{pless,tonchev}. The sum of the rows of $(I\mid S)$
is $c$, so $C_3$ contains the all-one word. We shall also use the
following consequence of the complete weight enumerator of an extremal
ternary self-dual $[48,24,15]$ code containing that word:
\begin{equation}\label{eq:binary-words}
 \sum_{b\in C_3\cap\{0,1\}^{48}} t^{\wt(b)}
   =1+94t^{24}+t^{48}.
\end{equation}
This is the two-symbol specialization of
\cite[equation~(20)]{munemasa-tamura}. These two code-theoretic facts
determine which of the possible short vectors can occur in the even
neighbour.

\begin{theorem}[The mother shell]\label{thm:p48-shell}
The lattice $L'$ in~\eqref{eq:characteristic} is even unimodular,
has minimum squared norm $6$, and has exactly $52416000$ vectors of
squared norm $6$.
\end{theorem}
\begin{proof}
Evenness and unimodularity were established above. Suppose first that
$z/\sqrt3\in L_0$ has squared norm $2$ or $4$. Then $z\cdot z$ is
$6$ or $12$. A nonzero residue $z\bmod3$ would have Hamming weight
at least $15$, whereas its weight is at most $z\cdot z$. Hence
$z=3w$ for some integer vector $w$. Its squared norm is $3w\cdot w$,
which is even only when $w\cdot w$ is even, and is therefore at least
$6$ if nonzero.

A vector in the other coset has the form $y/(2\sqrt3)$, where
$y=g+2z$ and $z/\sqrt3\in L_0$. All $48$ coordinates of $y$ are odd,
so its squared norm is at least $4$. Equality would force
$y\in\{1,-1\}^{48}$. Let $b$ be the incidence vector of its negative
coordinates and put $w=\wt(b)$. Since $g\equiv\boldsymbol1\pmod3$,
the identity $z=(y-g)/2$ gives $-z\equiv b\pmod3$. Consequently
$b\in C_3$, and~\eqref{eq:binary-words} implies $w\in\{0,24,48\}$.
On the other hand, $c_i\equiv1\pmod6$, $\sum_i c_i=24$, and
$c\cdot g=582$. It follows that
\begin{equation}\label{eq:odd-parity}
 c\cdot z
 =\frac{\sum_i c_i-c\cdot g}{2}-\sum_{i:b_i=1}c_i
 \equiv3-w\equiv3\pmod6.
\end{equation}
This contradicts $z/\sqrt3\in L_0$, which requires
$c\cdot z\equiv0\pmod6$. Thus neither coset contains a vector of
squared norm $2$ or $4$. The vector $\sqrt3(e_0+e_1)\in L_0$
has squared norm $6$, proving the asserted minimum.

It remains to count the minimum shell. Write
$\Theta_{L'}(q)=\sum_{x\in L'}q^{x\cdot x/2}$. The theta series of
an even unimodular lattice of rank $48$ is a modular form of weight
$24$~\cite{conway1999sphere,mos}. The space has basis
$E_4^6,E_4^3\Delta,\Delta^2$. Its constant term and the vanishing
coefficients of $q$ and $q^2$ therefore give
\begin{equation}\label{eq:theta48}
 \Theta_{L'}=E_4^6-1440E_4^3\Delta+125280\Delta^2.
\end{equation}
Substituting
$E_4=1+240q+2160q^2+6720q^3+O(q^4)$ and
$\Delta=q-24q^2+252q^3+O(q^4)$ yields
\[
 \Theta_{L'}=1+52416000q^3+O(q^4),
\]
as required.
\end{proof}

Within this Pless-code model, the shift can equivalently be written
using $g_*=\boldsymbol1+6e_0$: indeed,
$(g-g_*)/2=-12e_{24}\in\sqrt3L_0$, so $g$ and $g_*$ define the
same coset of $L_0$. This is the all-one neighbour construction in
\cite{munemasa-tamura}; the parity calculation above explains directly
why this choice excludes the norm-four vectors. The liftings of
Section~\ref{app:p48} use the catalogue Gram basis and its published
minimum directly; the projection of Section~\ref{sec:projection45}
uses the separately specified QR neighbour.

The norm-six shell of $L'$, divided by $\sqrt6$,
is a kissing configuration $C$ of size $52416000$. Indeed, for two
distinct shell vectors $x,x'$, the nonzero lattice vector $x-x'$
has squared norm at least six, so $x\cdot x'\le3$.

\section{Flexibility of auxiliary configurations}\label{sec:deformations}

The layer inequalities identify which constraints are tight and which
leave room for deformation. A \emph{contact} is a pair of distinct
points with inner product $1/2$. Here $X_{25}$ is the 197058-point
endpoint model of~\eqref{eq:d25}, and $X_T,X_M,X_R$ are the ERS layers
of Section~\ref{sec:codes}. The dimension-38 observation concerns
that layered-code construction; the 591900-point construction instead
uses the Leech lift of Section~\ref{sec:completion}.
The connection between contact
constraints and rigidity of spherical codes is developed
in~\cite{cjkt}; here it can be made quantitative directly from the
proved inner-product margins.

\subsection{Paired motion in dimension 25}
Write $s\in S_*/\sqrt{32}$ and
$x_s^{\pm}=(\sqrt3\,s/2,\pm1/2)$, with poles $p^{\pm}=(0,\pm1)$.
The contacts involving these new points are exactly the $496$ paths
\begin{equation}\label{eq:contact-paths}
 p^+\;--\;x_s^+\;--\;x_s^-\;--\;p^-.
\end{equation}
Their internal vertices are disjoint. They have no contacts with the
retained equator, because the relevant inner products are at most
$\sqrt3/4<1/2$. Between different directions the same-latitude bound
is $7/16$, and the opposite-latitude bound is $-1/16$. Thus the
displayed graph has $994$ vertices and $1488$ edges.

\begin{proposition}\label{prop:d25-motion}
Keep the equator and the two poles of $X_{25}$ fixed. For each
$s\in S_*/\sqrt{32}$ choose independently a unit vector $\widehat s$
with $\norm{\widehat s-s}<1/24$. Replacing $x_s^\pm$ by
$(\sqrt3\,\widehat s/2,\pm1/2)$ preserves the kissing condition and
the contacts~\eqref{eq:contact-paths}.
\end{proposition}
\begin{proof}
Let $\delta=\max_s\norm{\widehat s-s}<1/24$. For unit vectors,
\[
 \bigl|\ip{\widehat s}{\widehat t}-\ip{s}{t}\bigr|
 \le\norm{\widehat s-s}+\norm{\widehat t-t}\le2\delta.
\]
Consequently distinct same-latitude directions have inner product
less than $7/16+3\delta/2<1/2$, and distinct opposite-latitude
directions have inner product less than $-1/16+3\delta/2<0$.
For a fixed equatorial direction $z$,
\[
 \frac{\sqrt3}{2}\ip{z}{\widehat s}
 \le\frac{\sqrt3}{4}+\frac{\sqrt3}{2}\delta
 <\frac{13\sqrt3}{48}<\frac12.
\]
The last inequality follows from $3\cdot13^2<24^2$. A paired pair
still has inner product $1/2$, as do its contacts with the appropriate
poles. All other pole products are unchanged.
\end{proof}

With the equator, poles, and labels fixed, the common direction of
each pair is free to move in an open subset of $S^{23}$. The height is
forced, but the direction is not.
Thus saturation of the latitude constraints does not imply rigidity
of the lifted configuration.

\subsection{Independent motion of the dense layers}
In dimensions $m=n=38,39$, the dense sign codes have minimum distance
ten. Their mutual inner products are therefore at most $1-20/m<1/2$.
Their products with the middle and root layers are at most
$\sqrt{8/m}$ and $\sqrt{2/m}$, also strictly less than $1/2$.
Every dense point is an isolated vertex of the contact graph.

\begin{proposition}\label{prop:dense-motion}
For $m=n\in\{38,39\}$, keep $X_M$ and $X_R$ fixed and put
\[
 \eta_m=\frac{40-m}{4m}.
\]
Each $x\in X_T$ may independently be replaced by a unit vector
$\widehat x$ with $\norm{\widehat x-x}<\eta_m$. The resulting
configuration has the same cardinality and satisfies the kissing
condition. No new contact involves a dense point.
\end{proposition}
\begin{proof}
The change in a product of two unit vectors is at most the sum of
their displacement lengths. Hence distinct moved dense points have
inner product strictly less than
\[
 1-\frac{20}{m}+2\eta_m=\frac12.
\]
For a moved dense point and a fixed middle point, the bound is
$\sqrt{8/m}+\eta_m<1/2$. To verify the latter inequality exactly,
the positive number $1/2-\eta_m$ has square exceeding $8/m$ by
$153/5776$ at $m=38$ and by $937/24336$ at $m=39$.
The root-layer bound is smaller. The other pairs are unchanged.
\end{proof}

The respective displacement radii are $1/76$ and $1/156$.
The codes provide exact starting configurations, while their
neighbourhoods need not retain a binary description. The two motion
results exhibit distinct sources of flexibility: collective motion
of a contacting pair in dimension $25$, and independent motion of
strictly separated points in dimensions $38$ and $39$.

\clearpage
\bibliographystyle{unsrturl}
\bibliography{references}
\end{document}